\documentclass[11pt]{article}
\usepackage{randolph}
\usepackage[boxed,commentsnumbered,noend,linesnumbered]{algorithm2e}
\usepackage{xfrac}
\usepackage{lmodern}
\usetikzlibrary{trees, shapes, arrows.meta, positioning, calc}

\newcommand{\Aa}{\mathcal{A}}
\newcommand{\Bb}{\mathcal{B}}

\newcommand{\Oh}{O}
\newcommand{\Os}{\Oh^\ast}
\newcommand{\Prob}[1]{\mathbf{P} \left[ #1 \right]}

\newcommand{\diam}{\mathrm{diam}}
\let\oldnl\nl
\newcommand{\nlnonumber}{\renewcommand{\nl}{\let\nl\oldnl}}

\newcommand{\pCert}{\mathcal{Q}} 
\newcommand{\inputA}{A} 
\newcommand{\inputB}{B} 
\newcommand{\auxL}{\Aa} 
\newcommand{\auxR}{\Bb} 

\newenvironment{claimproof}[1]{\par\noindent Proof:\space#1}{\hfill $\blacksquare$}

\newcommand{\newreptheorem}[2]{%
\newenvironment{rep#1}[1]{%
 \def\reptitle{#2 \ref{##1}}%
 \begin{reptheorem}}%
 {\end{reptheorem}}}
 \makeatletter
\newsavebox\myboxA
\newsavebox\myboxB
\newlength\mylenA

\newcommand*\xoverline[2][0.75]{%
    \sbox{\myboxA}{$\m@th#2$}%
    \setbox\myboxB\null
    \ht\myboxB=\ht\myboxA%
    \dp\myboxB=\dp\myboxA%
    \wd\myboxB=#1\wd\myboxA
    \sbox\myboxB{$\m@th\overline{\copy\myboxB}$}
    \setlength\mylenA{\the\wd\myboxA}
    \addtolength\mylenA{-\the\wd\myboxB}%
    \ifdim\wd\myboxB<\wd\myboxA%
       \rlap{\hskip 0.5\mylenA\usebox\myboxB}{\usebox\myboxA}%
    \else
        \hskip -0.5\mylenA\rlap{\usebox\myboxA}{\hskip 0.5\mylenA\usebox\myboxB}%
    \fi}
\makeatother
 
\newreptheorem{theorem}{Theorem}
\newreptheorem{lemma}{Lemma}

\title{ 
Beating Meet-in-the-Middle for Subset Balancing Problems 
}

\author{ Tim Randolph\footnote{Harvey Mudd College,
    \texttt{trandolph@hmc.edu}.} \and Karol W\k{e}grzycki\footnote{Max Planck Institute for Informatics,
        Saarbrücken, Germany, \texttt{kwegrzyc@mpi-inf.mpg.de}. Supported by the Deutsche Forschungsgemeinschaft (DFG, German Research Foundation) grant number
     559177164. 
 }
 }
\date{}

\begin{document}

\maketitle

\begin{abstract}
    We consider exact algorithms for Subset Balancing, a family of related problems that generalizes Subset Sum, Partition, and Equal Subset Sum. Specifically, given as input an integer vector $\vec{x} \in \Z^n$ and a constant-size coefficient set $C \subset \Z$, we seek a nonzero solution vector $\vec{c} \in C^n$ satisfying $\vec{c} \cdot \vec{x} = 0$.

    For $C = \{-d,\ldots,d\}$, $d > 1$ and $C = \{-d,\ldots,d\}\setminus\{0\}$, $d > 2$, we present algorithms that run in time $O(|C|^{(0.5 - \epsilon)n})$ for a constant $\epsilon > 0$ that depends only on $C$. These are the first algorithms that break the $O(|C|^{n/2})$-time ``Meet-in-the-Middle barrier'' for these coefficient sets in the worst case. This improves on the result of Chen, Jin, Randolph and Servedio (SODA 2022), who broke the Meet-in-the-Middle barrier on these coefficient sets in the average-case setting. We also improve the best exact algorithm for Equal Subset Sum (Subset Balancing with $C = \{-1,0,1\}$), due to Mucha, Nederlof, Pawlewicz, and W\k{e}grzycki (ESA 2019), by an exponential margin. This positively answers an open question of Jin, Williams, and Zhang (ESA 2025).  Our results leave two natural cases in which we cannot yet break the Meet-in-the-Middle barrier: $C = \{-2, -1, 1, 2\}$ and $C = \{-1, 1\}$ (Partition). 

    Our results bring the representation technique of Howgrave-Graham and Joux (CRYPTO 2010) from average-case to worst-case inputs for many $C$. This requires a variety of new techniques: we present strategies for (1) achieving good ``mixing'' with worst-case inputs, (2) creating flexible input representations for coefficient sets without 0, and (3) quickly recovering compatible solution pairs from sets of vectors containing ``pseudosolution pairs''. These techniques may find application to other algorithmic problems on integer sums or be of independent interest.
\end{abstract}

\thispagestyle{empty}
\clearpage
\setcounter{page}{1}
\section{Introduction}
\label{sec:intro}

In the Subset Sum problem, we are given a set of $n$ integers and want to partition
it into two subsets of equal sum.\footnote{An alternative formulation searches for a subset of the input that sums to a specified target $t$. Our (equivalent) definition of the problem is often called \emph{Partition}. See~\cite{mfcs12} for a discussion of more equivalent problems such as Knapsack.} This problem and its variants are among the most famous \NP-complete problems and are central to the study of algorithmic problems on integers. 

In the early 1970s, Horowitz and Sahni~\cite{horowitz1974computing} designed an exact algorithm for Subset Sum using a strategy now known as ``Meet-in-the-Middle''. A careful implementation of this algorithm runs in $O(2^{0.5n})$ time, without additional polynomial factors \cite{chen2023logshavingSS}. Despite many attempts to improve on this algorithm, Meet-in-the-Middle remains the standard for exact algorithms for Subset Sum, and whether a $\Oh(2^{(0.5-\eps)n})$-time algorithm for Subset Sum exists for any constant $\eps > 0$ is one of the major open questions in the field.

The last decade has witnessed a surge of research activity focused on exact algorithms for the Subset Sum problem, including several promising results closely related to this open question.
In a breakthrough result, Howgrave-Graham and
Joux \cite{howgrave2010new} broke the Meet-in-the-Middle barrier for Subset Sum
in the \emph{average-case} setting with a $\Oh(2^{0.337n})$-time algorithm. Their
work introduced what is now known as the ``representation technique''
(see \Cref{subsec:rep-method}). Subsequent work along these lines has refined the technique and improved the average-case runtime to $\Oh(2^{0.283n})$ \cite{becker2011improved,bonnetain2020improved}.  

Progress on Subset Sum has also spurred research on subset balancing problems, a category that generalizes Subset Sum. For example,
Mucha, Nederlof, Pawlewicz and Węgrzycki considered the problem of \emph{Equal
Subset Sum} (ESS): does a given set of integers contain any two
subsets with the same sum? Surprisingly, these authors showed that the
Meet-in-the-Middle barrier for Equal Subset Sum could be broken by an
exponential factor for \emph{worst-case} inputs \cite{mucha2019equal}.

\subsection{Subset Balancing}

Equal Subset Sum and Subset Sum are both special cases of the following
problem. (They result from setting $C = \{-1, 0, 1\}$ and $\{-1, 1\}$, respectively.)

\begin{prob}[Subset Balancing]{prob:SB}
    \textbf{In:} A vector $\vec{x} \in \Z^n$ and a fixed set $C \subseteq \Z$. \\
    \textbf{Out:} A solution vector $\vec{c} \in C^n \setminus \{\vec{0}\}$ satisfying $\vec{c}
    \cdot \vec{x} = 0$.
\end{prob}

Subset Balancing was previously studied by Chen, Jin, Randolph and Servedio~\cite{chen2022average}. We follow these authors in considering two families of coefficient sets which generalize Subset Sum and Equal Subset Sum, respectively:
\begin{itemize}
    \item $[-d : d] \coloneqq \{-d, -d+1, \dots, d-1, d\}$, integers with absolute value at most $d$ \emph{including} $0$; and
    \item $[\pm d] \coloneqq \{-d, \dots, -1, 1, \dots, d\}$, integers with absolute value at most $d$ \emph{excluding} $0$.
\end{itemize}

Subset Balancing can be solved in time $\Os(|C|^{0.5n})$\footnote{$\Os$ notation suppresses $\poly(n)$ factors.} on any coefficient set
by using the (suitably generalized) Meet-in-the-Middle algorithm: partition the
input integers into two sets of size $0.5n$, then enumerate lists of all $O(|C|^{0.5n})$ possible coefficient-weighted sums for each half. Then, sort the lists and search for a matching pair of sums. 

\cite{chen2022average} proposed the following question regarding exact algorithms for Subset Balancing:

\begin{question}[Meet-in-the-Middle Barrier for Subset Balancing] \label{oq:mim-barrier2}
    Can Subset Balancing be solved in time $\Oh(|C|^{(0.5-\epsilon)n})$ for a constant $\epsilon > 0$?
\end{question}
These authors generalized the representation technique to achieve $\Oh(|C|^{(0.5-\epsilon)n})$-time algorithms for many coefficient sets in the \emph{average-case} setting, in which inputs are sampled uniformly at random from an exponential range. The natural next question is whether similar improvements can be achieved in the more challenging \emph{worst case}.

\subsection{ Our Results } 

We establish faster algorithms for Subset Balancing on almost every coefficient
set of the form $[-d:d]$ and $[\pm d]$. All of our algorithms are Monte Carlo
algorithms that never return false positives and whose success probability can
be amplified to $1-2^{-\Omega(n)}$ by repetition. Specifically, we present the following results:

\begin{description}
    \item \textbf{Result 1:} Worst-case algorithms for Subset Balancing on all
        coefficient sets $C = [-d:d]$, $d > 1$ that run in time
        $\Oh(|C|^{(0.5 - \epsilon)n})$ for a constant $\epsilon > 0$ that
        depends only on $d$ (\Cref{thm:main-with0}). 
\end{description}

Result 1 directly addresses~\cref{oq:mim-barrier2}. These are the first algorithms to break the Meet-in-the-Middle barrier on these coefficient sets in the worst case, and represent a similar accomplishment to \cite{chen2022average} in a strictly harder setting. This requires overcoming a technical hurdle: in the average case, solution representations ``mix well'' with high probability over the input, which allows for the representation method to be applied (\Cref{subsec:rep-method}). However, this condition does not hold in the worst case. We formalize the idea of a ``mixing dichotomy'' to address this problem (\Cref{subsec:mixing-dichotomies}).

\begin{description}
    \item\textbf{Result 2:} Worst-case algorithms for Subset Balancing on
        all coefficient sets $C = [\pm d]$, $d > 2$ that run in time
        $\Oh(|C|^{(0.5 - \epsilon)n})$ for a constant $\epsilon > 0$ that depends
        only on $d$ (\Cref{thm:main-without0}). 
\end{description}

These are likewise the first algorithms that break the Meet-in-the-Middle barrier on coefficient sets of the form $[\pm d]$ in the worst-case regime. This result demonstrates that we can achieve better algorithms for coefficient sets beyond contiguous intervals, including those on the pattern of Subset Sum ($C = [\pm 1]$). We separate this result from Result 1 because it requires additional insight: specifically, a more expansive notion of a solution representation based on sumsets (\Cref{subsec:coeff-shifting-techoverview}).

\begin{description}
    \item \textbf{Result 3:} A new worst-case algorithm for Equal Subset
        Sum (Subset Balancing on $C = [-1 : 1]$) that runs in time
        $\Oh(1.7067^n)$-time (\Cref{thm:ess}).
\end{description}

This improves upon the result of Mucha, Nederlof, Pawlewicz and W\k{e}grzycki~\cite{mucha2019equal}, who gave an $\Oh(1.7088^n)$-time algorithm. In addition to the above techniques, this result uses a generalization of the recently invented method of ``compatibility certificates'' (\cite{NederlofW21}) and represents the first time this technique has been used to improve the \emph{running time} of an algorithm as opposed to the space complexity.

Improving the exact algorithm for Equal Subset Sum has recently gained attention in the context of work on closely related problems such as Pigeonhole Equal Subset Sum \cite{DBLP:conf/esa/AllcockHJKS22,DBLP:conf/icalp/0001W24,jin2025new}. Result 3 demonstrates that additional improvements beyond \cite{mucha2019equal} are possible, positively answering an open question of \cite{jin2025new} (Section 5, Question 1). 

In summary, we improve upon the Meet-in-the-Middle algorithm for a wide variety of Subset Balancing problems, with two important exceptions: (1) the Subset Sum/Partition problem
itself and (2) Subset Balancing on $C = [\pm 2]$.\footnote{ Note that (1) would imply (2). Namely, to solve an instance $(a_1,\ldots,a_n)$ of Subset Balancing on $[\pm 2]$, we could use a faster than Meet-in-the-Middle algorithm for Subset Balancing on $[\pm 1]$ on an instance $(a_1/2,\ldots,a_n/2, 3a_1/2,...,3a_n/2)$.}
A pessimistic interpretation of our results might be that if (2) is difficult, additional insights are needed to beat Meet-in-the-Middle for Subset Sum. On the other hand, our results imply that the barrier-breaking result of \cite{mucha2019equal} for Equal Subset Sum is not specific to $C = [-1 : 1]$. An optimistic interpretation of our results is that they strengthen the circumstantial case that the Meet-in-the-Middle barrier can be broken for Subset Sum.

Our proofs adapt, generalize and draw inspiration from much of the scholarship on similar problems sketched in this introduction. They also require several new conceptual ingredients, overviewed in \Cref{sec:technical-overview}, that may find further application on similar problems.

\subsection{Related Work}
\label{subsec:related-work}

Subset Sum has seen numerous applications in
cryptography~\cite{crypto-merkle78}. In the low-density, average-case regime,
the problem can be solved in polynomial time~\cite{LO,Frieze,DBLP:journals/corr/abs-2408-16108}.

In the pseudopolynomial regime, potential improvements over the classic $O(nW)$-time dynamic programming algorithm are limited: we know that no $W^{1-\epsilon} \cdot \poly(n)$ algorithm is possible for any constant $\epsilon > 0$ assuming SETH~\cite{abboud2022seth}, where $W$ is the maximum absolute value of an input integer. Hence, researchers have focused on optimizing polynomial factors in $n$ (see
e.g.,~\cite{abboud2022seth,bringmann-soda,subset-sum-lower,sosa-wu,
koiliaris-soda,pseudopolynomial-polyspace}).

Subset Sum also appears more tractable if the input has significant additive structure. In \cite{austrin2015subset} and \cite{austrin2016dense}, Austrin, Kaski, Koivisto and Nederlof showed that Subset Sum instances with few subsets adding to any given integer (small ``bucket size'') or very many subsets adding to a certain integer (large ``bucket size'') could be solved in time $\Oh(2^{(0.5 - \eps)n})$.  Randolph and Węgrzycki showed that Subset Sum can be solved efficiently if the input set has a small \emph{doubling constant}, a condition that implies significant additive structure \cite{randolph2024parameterized}.

Still other recent work has made advancements related to space complexity, including better time-space trade-offs~\cite{dinur2012efficient,austrin2013space}, polynomial-space algorithms exponentially faster than $\Os(2^n)$ \cite{bansal,DBLP:conf/soda/ChenJWW22}, and algorithms that run in time $\Os(2^{0.5n})$ and use space $O(2^{(0.25 - \epsilon)n})$ \cite{NederlofW21,DBLP:conf/esa/BelovaCKM24}).

\subsection{Organization}

\begin{itemize}
    \item \Cref{sec:technical-overview} overviews the technical ingredients used in our results and the logical flow of our algorithms.
    \item \Cref{sec:prelims} summarizes notation and states a useful technical lemma related to the representation technique.
    \item \Cref{sec:SB-unbalanced} formally proves the folklore result that instances of Subset Balancing that are ``$\epsilon$-unbalanced'' can be solved in time $\Os(|C|^{(0.5 - \delta)n})$ for a constant $\delta = \delta(|C|, \epsilon) > 0$ that depends only on $|C|$ and $\epsilon$. It introduces several definitions (solution profiles, unbalanced instances) and assumptions that hold without loss of generality (solutions always exist, solution profiles can be guessed) that are used to simplify arguments in later sections.
    \item \Cref{sec:SB-with0} proves Result 1, our main result for coefficient sets of the form $C = [-d : d]$, $d > 1$.
    \item \Cref{sec:SB-without0} proves Result 2, our main result for coefficient sets of the form $C = [\pm d]$, $d > 2$.
    \item \Cref{sec:ess} proves Result 3, our main result for Equal Subset Sum ($C = [-1 : 1]$).
    \item \Cref{sec:compatibility-test} extends the ``compatibility certificate'' method, first used in \cite{NederlofW21}. We use this result as an ingredient to achieve the desired runtime for the algorithm presented in \Cref{sec:ess}.
\end{itemize}

\section{Technical Overview}
\label{sec:technical-overview}

In this section we overview the conceptual ingredients used in our results. One (the representation technique) is familiar but needs an introduction to provide context for the others. Two (mixing dichotomies and coefficient shifting) are original, and two (representations admitting pseudosolutions and compatibility certificates) generalize ideas implicit in existing work. All require careful adaptation to work in our setting.

We conclude with a subsection that overviews the logical structure of our results.

\subsection{The Representation Technique}
\label{subsec:rep-method}

The representation technique was introduced by Howgrave-Graham and
Joux in~\cite{howgrave2010new} and used by Mucha, Nederlof, Pawlewicz and Węgrzycki in their algorithm for Equal Subset Sum~\cite{mucha2019equal} and by Chen, Jin, Randolph and Servedio in their algorithms for average-case Subset Balancing~\cite{chen2022average}, among others. In this work, we extend the technique further to worst-case Subset Balancing.

The key idea behind the representation technique is to implicitly represent a
single solution vector $\vec{c} \in C^n$ as a set of \emph{solution pairs} $(\vec{a}, \vec{b})$ satisfying $\vec{a} - \vec{b} = \vec{c}$.\footnote{It is sometimes more convenient to define solution pairs so that they satisfy $\vec{a} + \vec{b} = \vec{c}$; this is equivalent.} Roughly speaking, the representation technique has three steps:

\begin{enumerate}
    \item \textbf{Redefine the Search Space. } Define a large ambient set $Y$ with $|Y| > |C|^{n/2}$ containing many vectors, including the individual members $\vec{a}$ and $\vec{b}$ of each solution pair $(\vec{a}, \vec{b})$.
    \item \textbf{Filter the Search Space.} Filter the ambient set to create a smaller set $Z \subset Y$ with $|Z| < |C|^{n/2}$ that is likely to contain at least one solution pair. This is typically done by preserving the subset of vectors $\vec{y} \in Y$ that satisfy
    \[
        \vec{y} \cdot \vec{x} \equiv \bm{r} \pmod{\bp}
    \]
    for a fixed Subset Balancing instance $\vec{x}$, a randomly chosen prime $\bp$ and a randomly chosen residue class $\bm{r} \in [\bp]$. This produces a set $Z$ of size approximately $|Y| / \bp$ with the helpful property that for every solution pair $(\vec{a}, \vec{b})$, 
    \[
        \vec{a} \cdot \vec{x} \equiv \bm{r}\pmod{\bp} \implies \vec{b} \cdot
        \vec{x} \equiv \bm{r}\pmod{\bp}
    \]
    because
    \[
        (\vec{a} - \vec{b}) \cdot \vec{x} = \vec{c} \cdot \vec{x} = 0
    \]
    by the definition of a solution pair.
    \item \textbf{Recover Solutions.} Search within $Z$ to recover a solution pair. 
\end{enumerate}

As an example, consider \cite{mucha2019equal}'s application of the representation technique to the case of $C = \{-1, 0, 1\}$ (Equal Subset Sum). These authors begin with the observation that the worst-case runtime of existing Meet-in-the-Middle-style approaches occurs when an Equal Subset Sum instance $\vec{x}$ admits only ``balanced'' solutions: that is, solution vectors $\vec{c}$ with approximately equal numbers of $0$'s, $1$'s, and $-1$'s ($|\vec{c}^{\,-1}(0)| \approx |\vec{c}^{\,-1}(1)| \approx |\vec{c}^{\,-1}(-1)| \approx n/3$).

In this case, the set of candidate solution pairs $\{0, 1\}^n \times \{0, 1\}^n$ admits many pairs $(\vec{a}, \vec{b})$ satisfying $\vec{a} - \vec{b} = \vec{c}$ for any fixed solution vector $\vec{c}$. In particular, if an Equal Subset Sum instance $\vec{x}$ admits a solution $\vec{c}$ in which $n/3$ indices are 0, the number of solution pairs is $2^{n/3}$, as for every index $i$ with $\vec{c}_i = 0$ we can choose $(\vec{a}_i,\vec{b}_i)$ to be either $(0,0)$ or $(1,1)$. Accordingly \cite{mucha2019equal} filter the set of candidate solution pairs by choosing a random prime $\bp \approx 2^{n/3}$ and a random residue class $\bm{r} \in [\bp]$ and enumerating the set
\begin{align*}
    \calS \coloneqq & \{\vec{v} \in \{0, 1\}^n \mid \vec{v} \cdot \vec{x} \equiv \bm{r} \pmod{\bm{p}}\}.
\end{align*}
This choice of $\bp$ ensures both that $\calS$ contains a solution pair $(\vec{a}, \vec{b})$ satisfying $\vec{a} - \vec{b} = \vec{c}$ with sufficiently high probability and that
\[
    |\calS| \approx \frac{|\{0, 1\}^n|}{\bp} \approx 2^{2n/3},
\]
so this pair can be recovered in time $\Os(2^{2n/3})$ by enumerating and
searching $\calS$ in output-linear time. We omit some details, such as handling instances with ``unbalanced'' solution vectors (see \Cref{sec:SB-unbalanced}) and the reason why $\{0, 1\}^n \cdot \vec{x}$ can be assumed to produce $2^n$ distinct dot products (implicitly, a simple mixing dichotomy).

\subsection{ Mixing Dichotomies }
\label{subsec:mixing-dichotomies}

Crucially, the search space filtering step of the representation technique requires that
solution pairs fall into many different residue classes modulo $\bp$ when the
dot product is taken with an input vector $\vec{x}$. Otherwise, a residue class chosen at random
is unlikely to contain at least one solution pair. This occurs with high probability as long as solution pairs $(\vec{a}, \vec{b})$ create many different dot products $\vec{a} \cdot \vec{x}$, an event referred to as ``mixing well'' with respect to $\vec{x}$. When input numbers are drawn uniformly at random from a large enough range, solution vectors mix well with high probability over the random draw of the input, a fact
that enabled the first successful average-case algorithms based on the representation technique \cite{howgrave2010new,becker2011improved}.

However, solution pairs do not always mix well with respect to worst-case instances of Subset Balancing. To make the representation technique work in these cases, we will prove \emph{mixing dichotomies}: statements of the form ``if solution pairs do \emph{not} mix well with respect to a Subset Balancing instance $\vec{x}$, $\vec{x}$ has some other useful property (that allows us to break the Meet-in-the-Middle barrier using a different approach)''.

We prove the following mixing dichotomies:
\begin{itemize}
    \item \Cref{lem:perf-mixing}: if solution pairs do not mix well with respect to a Subset Balancing instance on a coefficient set $C = [-d : d]$, the instance admits an ``unbalanced'' solution (solution in which the frequency of coefficients deviates from uniform by $\Omega(n)$).
    \item \Cref{lem:perf-mixing2}: if solution pairs do not mix well with respect to a Subset Balancing instance on a coefficient set $C$ that can be ``coefficient-shifted'' (see the next subsection), the instance admits exponentially many distinct solutions.
    \item \Cref{lem:ess-perfect-mixing}: if our set of ``good solution pairs'' (a definition specific to our algorithm for Equal Subset Sum) do not mix well with respect to an ESS instance, the instance admits an ``unbalanced'' solution.
\end{itemize}

\subsection{ Coefficient Shifting: More Flexible Representations }
\label{subsec:coeff-shifting-techoverview}

When our coefficient set $C$ is a contiguous interval, it is relatively straightforward to represent a solution vector in $C^n$ as a solution pair. For example, a solution vector $\vec{c} \in [-1 : 1]^n$ can be represented in many ways as the difference of two vectors $\vec{a}, \vec{b} \in \{0, 1\}^n$ (as in \cite{mucha2019equal}). Likewise a solution vector $\vec{c} \in \{0, 1\}^n$ can be represented in many ways as a \emph{sum} of two vectors in $\{0, 1\}^n$, each with approximately half of the ``$1$'' indices in $\vec{c}$ (as in \cite{howgrave2010new}). $\vec{c} \in \{0, 1\}^n$ can even be represented as a sum of two vectors in $[-1 : 1]^n$, as long as the ``$-1$'' indices are carefully managed (as in \cite{becker2011improved}). 

However, when we generalize Subset Sum we face coefficient sets of the form $C = [\pm
d]$ (without the $0$ coefficient). Designing algorithms appears more challenging for these coefficient sets. For instance, breaking the Meet-in-the-Middle
barrier for $C = [-1 : 1]$ is comparatively straightforward, while breaking the
Meet-in-the-Middle barrier for $C = [\pm 1]$ remains a notorious open question. Chen, Jin, Randolph and Servedio also noted several difficulties that arise when handling coefficient sets of the form $[\pm d]$ that result from the lack of an $0$ coefficient \cite{chen2022average}. 

To address these issues we introduce the technique of \emph{coefficient shifting}, which effectively represents a coefficient set $C$ as the \emph{sumset} of two distinct and often differently-sized ``coefficient set factors'' $C_1$ and $C_2$. 
For example, the equation
\[
    \{0, 1\} + \{-3, -2, 1, 2\} = [\pm 3]
\]
allows us to design algorithms that represent solution vectors in $C^n = [\pm 3]^n$ as solution pairs chosen from $\{0, 1\}^n$ and $\{-3, -2, 1, 2\}^n$. Additionally, and crucially, it also creates multiple representations for $+2$ and $-2$ (as $0 + (-2)$ and $1+(-3)$, $0+2$ and $1 + 1$ respectively). \Cref{subsec:coeff-shifting} introduces the approach in more detail and subsequently applies it to break the Meet-in-the-Middle barrier for coefficient sets of the form $C = [\pm d], d > 2$.

\subsection{ Representations Admitting Pseudosolutions }
\label{subsec:to-reps-admitting-pseudosolutions}

One way to speed up the representation technique is to increase the number
of solution pairs by making the representation more flexible: specifically,
increasing the number of coefficients used to represent solution pairs. The idea
is to increase the \emph{ratio} of solution pairs to search space, ultimately
resulting in a smaller filtered search space. This idea originates from
\cite{becker2011improved}, who applied representations admitting pseudosolutions
to improve on \cite{howgrave2010new}'s algorithm for average-case Subset Sum.

In our approach to Equal Subset Sum, instead of considering solution pairs chosen from $\{0, 1\}^n$, we will enumerate the set
\begin{align*}
    \calS \coloneqq & \{\vec{v} \in [0:2]^n \mid |\vec{v}^{\,-1}(1)| = n/2, |\vec{v}^{\,-1}(2)| = \epsilon n, \vec{v} \cdot \vec{x} \equiv \bm{r} \pmod{\bm{p}}\}
\end{align*}
for a constant $\eps > 0$, a large random prime $\bp$ and a random residue class $\bm{r} \in [\bp]$. 

When $\eps$ is chosen carefully, $|\calS|$ promises to be small enough to
achieve a runtime improvement over \cite{mucha2019equal}, who used
solution pairs chosen from $\{0, 1\}^n$. However, there is a problem. Because
solution pairs contain ``2'' indices, it is not true that $\vec{a} - \vec{b} \in
[-1:1]^n$ for every $(\vec{a}, \vec{b}) \in \calS \times \calS$, even if $(\vec{a} -
\vec{b}) \cdot \vec{x} = 0$. Thus, in order to solve the problem, we need to
search $\calS$ for a ``compatible'' solution pair $(\vec{a}, \vec{b})$ with $\vec{a} - \vec{b} \in [-1:1]^n$. We consider a strategy for solving this problem quickly in the next subsection.

\subsection{ Compatibility Certificates for Coefficient Vectors }

A persistent challenge when designing an algorithm using the representation technique is
the specter of ``pseudo-solutions''. Often, choosing a set of solution pairs
that allow the search space to be efficiently filtered results in pairs of vectors that appear to solve the problem but create invalid coefficients. 
For example, if a solution vector $\vec{c} \in \{-1, 0, 1\}^n$ is represented by
pairs of vectors in $\{0, 1, 2\}^n$, the solution recovery step might erroneously return a pair $(\vec{a}, \vec{b})$ satisfying 
\[
    (\vec{a} - \vec{b}) \cdot \vec{x} = \vec{c} \cdot \vec{x} = 0,
\]
but with $\vec{a} - \vec{b}$ containing ``2'' indices and thus an invalid solution.

Such \emph{pseudosolution pairs} can be easily ruled out but naively the
process takes quadratic time in the length of the list(s) containing pseudosolutions. The problem boils down to answering the following question: ``Given two vector lists $\inputA$ and $\inputB$ and a coefficient set $C$, does there exist a vector pair $(\vec{a}, \vec{b}) \in \inputA \times \inputB$ such that $\vec{a} - \vec{b} \in C^n$?'' We prove the following result, which shows that this problem can be solved in near-linear time in our setting.

{
\renewcommand{\thetheorem}{\ref*{thm:compatibility}}
\begin{theorem}\label{thm:compatibility}
    Fix a constant $0 \le \eps \le 1/4$ and let $\inputA,\inputB \subseteq [0:2]^d$ be two sets of vectors such that for every $\vec{a} \in \inputA$ and $\vec{b} \in \inputB$ it holds that
    \begin{align*}
        |\vec{a}^{\,-1}(2)| = |\vec{b}^{\,-1}(2)| = \eps d \text{ and } |\vec{a}^{\,-1}(1)| =
        |\vec{b}^{\,-1}(1)| = d/2.
    \end{align*}
    There exists an algorithm that recovers $(\vec{a}, \vec{b}) \in \inputA \times \inputB$ such that $\vec{a}-\vec{b} \in [-1 : 1]^d$ with high probability,
    if such a pair exists, and runs in time $\Oh\left( 2^{c(\eps) d} \cdot \left(|\inputA| + |\inputB| \right)\right)$, where
    \begin{equation}
        \label{eq:c-certificate}
        c(\eps) \coloneqq
        (1-\eps) \cdot 
        H\left( \frac{\eps}{1-\eps} \right) +
        \left(1/2+\eps\right) \cdot H\left(
            \frac{4\eps}{1+2\eps}
        \right) - H(2\eps) +
     o(1).
    \end{equation}
\end{theorem}
\addtocounter{theorem}{-1}
    
}


We combine~\cref{thm:compatibility} with an approach for Equal Subset Sum that depends on representations admitting pseudosolutions (\Cref{subsec:to-reps-admitting-pseudosolutions}). Importantly, the $2^{c(\epsilon) n}$ factor is offset by the increase in speed gained from the more flexible representations, and the result is a faster algorithm.

To prove~\cref{thm:compatibility}, we build off a technique for solving sparse instances of Orthogonal Vectors due to~\cite{NederlofW21}. The idea is to sample \emph{compatibility certificates}, sets of indices that certify that a pair of vectors is compatible. In our case, a pair $(\vec{a}, \vec{b})$ could be certified by two sets of indices $L,R \subset [n]$ satisfying 
\[
    \vec{a}^{\,-1}(2) \subseteq L \subseteq \vec{b}^{\,-1}(1) \text{ and } \vec{b}^{\,-1}(2) \subseteq R \subseteq \vec{a}^{\,-1}(1)
\]
a property that guarantees $\vec{a}$ and $\vec{b}$ do not share a ``0'' and ``2'' at the same index and hence that $\vec{a} - \vec{b} \in [-1:1]^n$. (Although it is possible to achieve a result similar to \Cref{thm:compatibility} with a simpler bucketing approach, we are not aware of a simpler method that achieves a strong enough dependence on $\epsilon$ to improve on \cite{mucha2019equal}.)


\subsection{Combining The Techniques}



\begin{figure}[h!]
    \centering
    \includegraphics{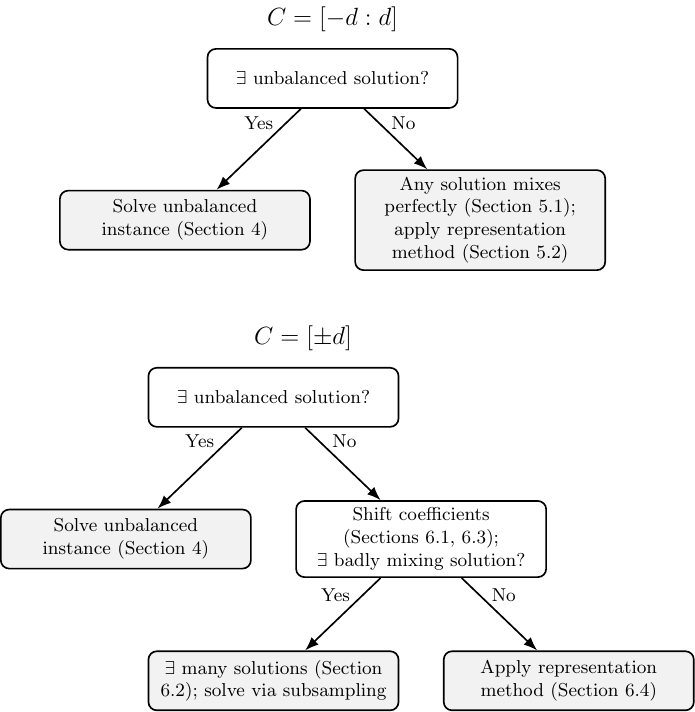}
    
    \caption{High-level overview of the logical flow of our algorithms.}
    \label{fig:logical-flow}
\end{figure}

We conclude the technical overview with a high-level summary of the logical flow of our algorithms (see \Cref{fig:logical-flow} for reference).

For any coefficient set $C$, we can efficiently guess the frequency of each coefficient $c \in C$ and choose an algorithmic strategy that searches for solutions matching this ``solution profile''. If the solution profile is unbalanced, we can apply a folklore variant of Meet-in-the-Middle to search for a matching solution (\Cref{sec:SB-unbalanced}).

If the solution profile is approximately balanced, we take a different approach according to $C$. For $C = [-d : d]$, $d > 1$, we apply a mixing dichotomy to conclude that if no unbalanced solution exists, any approximately balanced solution must mix perfectly (\Cref{subsec:perf-mixing-dichotomy}). We then apply a generalization of the representation method (\Cref{subsec:sb-repmethod-with0}). 

For $C = [\pm d]$, $d > 2$, we choose a pair of coefficient set factors (\Cref{subsec:coeff-shifting,subsec:num-soln-pairs-without0}) and apply a different mixing dichotomy to conclude that if any solution mixes poorly, there must exist exponentially many solutions; we can solve this case via a subsampling approach (\Cref{subsec:sb-without0-mixing}). Otherwise, any solution must mix well, in which case we use a generalized representation method approach (\Cref{subsec:sb-without0-repmethod}).

Finally, if $C = [-1 : 1]$, the approach is similar to other coefficient sets of the form $C = [-d : d]$. However, in this case we choose a solution representation that creates pseudosolutions in order to improve runtime (\Cref{subsec:ess-goodpairs}); this requires us to use the compatibility certificates method (\Cref{sec:compatibility-test}) and a special mixing dichotomy.

\paragraph{Further Runtime Improvements. }

In this work we focus on demonstrating that the Meet-in-the-Middle barrier can be broken for Subset Balancing on many coefficient sets, not on optimizing constant factors in the exponent of the running time. We have attempted to present our framework in a modular way so that the various parts can be adapted and extended, as integrating the various techniques more tightly might lead to additional improvements in running time. For example, we believe it is possible to extend the compatibility certificate methods we apply to the $C = [-1:1]$ case to further improve runtimes on $C = [-d:d]$ and $C = [\pm d]$. Likewise, it may be possible to apply the idea from \Cref{lem:perf-mixing2} to sharpen the runtime of the algorithm for $C = [\pm d], d > 1$: if there are exponentially many collisions between solution pairs corresponding to a nearly balanced solution, this should imply the existence of exponentially many \emph{unbalanced} solutions, a case which could be solved faster than instances with one unbalanced solution or with many balanced solutions. As the calculations are somewhat involved, we leave these optimizations for future work.

\section{Preliminaries}
\label{sec:prelims}

\textbf{Big-$O$ Notation.} We use standard big-$O$ notation with the addition of $O^*$, $\Omega^*$, and $\Theta^*$, which suppress polynomial functions of $n$. For example, $O^*(2^n) \leq 2^n \cdot n^{O(1)}$ and $\Omega^*(2^n) \geq 2^n \cdot n^{-O(1)}$. We use $\Theta^*(f(n))$ to indicate $O^*(f(n)) \cap \Omega^*(f(n))$.

\vspace{0.3cm}
\noindent
\textbf{Probability.} Variables written in \textbf{bold face} are random.\footnote{We use this convention to emphasize variables that are sampled directly but do not apply it everywhere for the sake of readability.
} Given a set $S$, we write $\bm{x} \sim S$ to indicate that $\bm{x}$ is an element sampled uniformly at random from $S$.

Our algorithms run in exponential time and admit one-sided error (i.e., no false positives). If such an algorithm recovers a solution with inverse polynomial probability (probability $n^{-O(1)}$), we can amplify the probability of correctness to a constant, and then further to $1 - e^{-\poly(n)}$, by repeating the algorithm $\poly(n)$ times. Amplifying the success probability in this way does not affect the runtime of our algorithms as written because the runtime statements use $\Os$ notation. When an algorithm satisfies this guarantee on correctness, we say it is correct ``with high probability.''

\vspace{0.3cm}
\noindent
\textbf{Sets.} We write $[n]$ for the integer set $\{1, 2, \dots, n\}$ and
$[a:b]$ (with $a \leq b$) for the integer set $[a, a+1, a+2, \dots, b]$. Because we will frequently refer to sets of the form $[-d : d] \setminus \{0\}$, we abbreviate this to $[\pm d]$. 

\vspace{0.3cm}
\noindent
\textbf{Set Operations.} Given integer sets $X$ and $Y$, we write:
\begin{itemize}
    \item $X + Y$ for the sumset $\{ a + b \; | \; a \in X, b \in Y\}$,
    \item $X \pmod{p}$ for $\{x \pmod{p} \; | \; x \in X\}$,
    \item $X = A \sqcup B$ to indicate that $A$ and $B$ are disjoint sets that partition $X$, and
    \item $\diam(X)$ for the diameter $\max_{x_1, x_2 \in X} |x_1 - x_2|$ of $X$.
\end{itemize}
We follow the general convention that for a function $f$ defined on the integers, $f(X)$ indicates the set obtained by applying $f$ element-wise to $X$.

\vspace{0.3cm}
\noindent
\textbf{Vectors.} Given an integer vector $\vec{x} \in \Z^d$ and integer $\alpha$, we write $\vec{x}^{\,-1}(\alpha)$ to denote the set containing each index $i \in [d]$ at which $\vec{x}_i = \alpha$.

Given a set of indices $I \subseteq [d]$, $\vec{v}_I$ denotes the $|I|$-dimensional restriction of $\vec{v}$ to the indices of $I$.

\vspace{0.3cm}
\noindent
\textbf{Memory Model.}
We make the assumption that standard operations on input integers (addition, multiplication, comparison, etc.) can be performed in constant time. For inputs of size $2^{n^{O(1)}}$, this can easily be reconciled with a more realistic RAM model: the running time of any operation that takes time $\poly(n)$ to execute on $O(n)$-bit integers is then absorbed by the $O^*$ notation.

\vspace{0.3cm}
\noindent
\textbf{Input Size.} We make the assumption that the integer inputs to Subset Balancing problems are bounded by $2^{n^{O(1)}}$. This is justified on both practical and theoretical grounds: first, inputs of size $2^{n^{\omega(1)}}$ require superpolynomial time to read and manipulate in any realistic memory model. Second, standard self-reductions based on hashing allow instances of Subset Balancing with extremely large integers to be effectively reduced to instances with integers of size $2^{O(n)}$ (see, e.g., \cite{randolph2024exact}, Section 2.3.4). 

\vspace{0.3cm}
\noindent
\textbf{Stirling's Approximation.} 
We use the following well-known consequence of Stirling's approximation to simplify binomial coefficients. For any constant $\alpha \in (0, 1)$, we have
\begin{equation}
    \label{eq:stirling}
    \binom{n}{\alpha n} = \Theta\left(n^{-1/2} \cdot 2^{H(\alpha) n}\right),
\end{equation}
where $H$ denotes the (binary) entropy function. A similar bound holds for
multinomial coefficients: given constants $\alpha_1,\alpha_2, \dots, \alpha_m$ summing to 1, we have
\begin{equation}
    \label{eq:stirling-multinomial}
    \binom{n}{\alpha_1 n, \; \alpha_2 n, \; \dots, \; \alpha_m n} =
    \Theta^*\left(2^{H(\alpha_1, \alpha_2, \dots, \alpha_m)n}\right),
\end{equation}
where in this case $H$ denotes the $m$-argument entropy function.

\subsection{ A Prime Distribution Lemma}
\label{subsec:helpful-lemmas}

During the filtering step of the representation technique (see \Cref{subsec:rep-method}), we will seek to recover at least one of a ``good set'' of integers $G$ hidden within a larger ``ambient set'' of integers $Y$. To do so, we will hash the elements of $Y$ modulo a large, randomly chosen prime $\bp$. Our intention is that both $G$ and $Y$ behave much as if we had randomly subsampled the elements of $Y$ with probability $1/\bp$: we want many residue classes $r \in [\bp]$ to contain approximately $|G| / \bp$ elements of $G$ and $|Y| / \bp$ elements of $Y$.

The following technical lemma makes this precise.

\begin{restatable}{lemma}{primedist}
    Fix an integer set $G$ with diameter $\diam(G) < 2^{n^{c}}$ for some constant $c > 0$. Let $Y \supseteq G$ be a superset of $G$, let $p_{max} \leq O(|G|)$
    be an integer bound of size $2^{\Omega(n)}$, and let $\bp \sim [p_{max} : 2p_{max}]$ be a prime. With constant probability, the set of ``good residue classes''
    \[
        \bm{R} \coloneqq \left\{ r \in [\bp] \;  \middle| \; 
            \left|\{g \in G \; | \; g \equiv r \pmod{\bp}\}\right| \geq
            \frac{|G|}{4p_{max}}, \; |\{y \in Y \; | \; y \equiv r \pmod{\bp}\}| \leq \frac{|Y|}{p_{max}} n^{c+1}
        \right\}
    \]
    has cardinality $|\bm{R}| > p_{max} \cdot n^{-(c+1)}$.
    \label{lem:prime-dist-main}
\end{restatable}

We prove \Cref{lem:prime-dist-main} in \Cref{apx:prime-dist-main}. In simple terms, a ``good residue class'' is one that contains $\Omega(|G| / \bp)$ elements of $G$ and $O^*(|Y| / \bp)$ elements of $Y$. \Cref{lem:prime-dist-main} says that with constant probability at least an $(n^{-O(1)})$-fraction of all residue classes are ``good''.

We will apply this lemma to coefficient-weighted inputs to the Subset Balancing problem. $G$ and $Y$ will be lists of exponential size containing sums representing coefficient-weighted partial solutions. We are only interested in cases in which $p_{max} = 2^{\Omega(n)}$; otherwise, this hashing procedure cannot give us exponential time savings. The requirement that $\diam(G) < 2^{n^c}$ for some constant $c$ follows from our input size assumption. The final requirement is that $p_{max} \leq O(|G|)$; often, we will set $p_{max} = \Theta(|G|)$ to optimize runtime.

\section{Folklore Results on Unbalanced Instances}
\label{sec:SB-unbalanced}

If a Subset Balancing instance admits a solution vector $\vec{c}$ that is ``unbalanced'' in the sense that some coefficients are more common than others, the instance can be solved faster than $|C|^{0.5n}$ using a fairly straightforward modification of Meet-in-the-Middle. This result is folklore: versions for Subset Sum and Equal Subset Sum are included in previous works (e.g., \cite{austrin2016dense,mucha2019equal,NederlofW21}) and the generalization to Subset Balancing is implicit in others (e.g., \cite{chen2022average}). We provide a formal proof because the result is a subroutine for our later algorithms and because the proof presents an opportunity to introduce several useful definitions.

\begin{definition}[Solution Profile]
    Fix a coefficient set $C$ and let $\vec{c} \in C^n$ be a coefficient vector. The \emph{solution profile} $\pi: C \rightarrow \N$ of $\vec{c}$ is the function
    \[
        \pi(z) = |\vec{c}^{\,-1}(z)|
    \]
    for each coefficient $z \in C$.
\end{definition}

In other words, $\pi$ counts the number of occurrences of each coefficient in the vector $\vec{c}$. There are only $O(n^{|C|})$ possible solution profiles for coefficient vectors of length $n$ supported on any fixed coefficient set $C$. Therefore, because our Subset Balancing algorithms will not return false positives we can effectively ``guess'' the correct solution profile by separately attempting to recover solutions matching each possible $\pi$. This will increase the overall runtime of our algorithms by a $\poly(n)$ factor that will ultimately be absorbed in the $\Os$ notation.

\begin{definition}[$\epsilon$-unbalanced solutions]
    We say that a solution vector $\vec{c} \in C^n$ with solution profile $\pi$ is $\epsilon$\emph{-unbalanced} if
    \[
        \left|\pi(c) - \frac{n}{|C|}\right| > \epsilon n
    \]
    for at least one coefficient $c \in C$.

    We refer to a solution vector that is \emph{not} $\epsilon$-unbalanced as $\epsilon$\emph{-balanced}, and to a solution vector that is $0$-balanced as \emph{perfectly balanced}.
\end{definition}






\begin{lemma}[Subset Balancing on $\epsilon$-Unbalanced Instances]\label{lem:unbalanced}
     For every constant $\eps > 0$ there exists $\delta > 0$ and an algorithm that solves Subset Balancing instances with $\eps$-unbalanced solutions, with high probability, in time
     \[
        \Os\left(|C|^{(1/2 - \delta)n}\right).
     \]
\end{lemma}


The algorithm behind~\Cref{lem:unbalanced} is a modified Meet-in-the-Middle approach. Our algorithm will not return false positives, so without loss of generality we assume the input instance has at least one solution and that we have correctly guessed the solution profile $\pi$ of a certain solution $\vec{c}$.

We will use the following set of coefficient vectors:
\begin{equation}
    \label{eq:C_a-def}
    \calS \coloneqq \calS(C, \pi) = \left\{ \vec{v} \in C^{n/2} \; \middle| \;
    |\vec{v}^{\,-1}(z)| = \pi(z)/2 \text{ for all } z \in C \right\}.
\end{equation}
$\calS$ contains the $\frac{n}{2}$-dimensional coefficient vectors supported
on $C$ that contain each coefficient in the same proportions as in
$\pi$.\footnote{ Here we make the assumption that both $n$ and $\pi(z)$, $z \in
C$ are even. If this is not the case, the algorithm is easily modified (one way
of doing so creates two sets $\calS_1$ and $\calS_2$ of slightly uneven size). Because these changes do not materially affect the algorithm, we leave them out to minimize clutter. }

Algorithm \ref{alg:UnbalancedSB} displays a Meet-in-the-Middle-style algorithm that takes as input a Subset Balancing instance $\vec{x}$, a coefficient set $C$ and a solution profile $\pi$ and attempts to recover a solution $\vec{c} \in C^n$ matching $\pi$. 

\begin{algorithm}[ht!]
    \DontPrintSemicolon
    \nlnonumber\textbf{function} $\textsf{UnbalancedSB}(\vec{x}, C, \pi)$:\\
    Select a partition $A \sqcup B$ of $[n]$ such that $|A| = |B| = n/2$ uniformly at random.\label{ln:unbalanced1}\\
    Enumerate $\calS \coloneqq \calS(C, \pi)$ as defined in \eqref{eq:C_a-def}.\label{ln:unbalanced2}\\
    Enumerate and sort the lists $L_A \coloneqq \calS \cdot \vec{x}_A$ and $L_B
    \coloneqq
    \calS \cdot \vec{x}_B$.\label{ln:unbalanced3}\\
    Use Meet-in-the-Middle to search for a solution pair in $L_A \times
    L_B$.\label{ln:unbalanced4}
    \caption{An algorithm for unbalanced instances of Subset Balancing.}
    \label{alg:UnbalancedSB}
\end{algorithm}

\begin{proof}[Proof of \cref{lem:unbalanced}: Correctness] Consider an instance $\vec{x}$ of Subset Balancing on $C$ that admits an $\epsilon$-unbalanced solution $\vec{c}$ for some $\epsilon > 0$. In order to recover $\vec{c}$, we repeat Algorithm \ref{alg:UnbalancedSB} for each of the $\poly(n)$ possible $\epsilon$-unbalanced solution profiles. Because Algorithm \ref{alg:UnbalancedSB} does not return false positives, we proceed under the assumption that we have correctly guessed the solution profile $\pi$ matching $\vec{c}$.

Given $\pi$, if Algorithm \ref{alg:UnbalancedSB} guesses a partition $A \sqcup
B$ in Line~\ref{ln:unbalanced1} such that $\vec{c}_A, \vec{c}_B \in \calS$,
Lines \ref{ln:unbalanced2}-\ref{ln:unbalanced4} recover $\vec{c}$ deterministically. Standard concentration inequalities imply that this event occurs with probability $n^{-O(|C|)}$ over the choice of $A$.
This can be amplified to high probability as discussed in \Cref{sec:prelims}.
\end{proof}

\begin{proof}[Proof of \cref{lem:unbalanced}: Runtime]
Recalling that $\poly(n)$ factors incurred by repeating an algorithm are absorbed in the $\Os$ notation, we proceed to bound the runtime of Algorithm \ref{alg:UnbalancedSB} by $O(|C|^{(1/2 - \delta)n})$ for a constant $\delta$ that depends only on $\epsilon$.

Each of lines~\ref{ln:unbalanced2}-\ref{ln:unbalanced4} in Algorithm~\ref{alg:UnbalancedSB} takes time
\begin{align*}
    \label{eq:SB-wc-mim-runtime}
    \Os(|\calS|) &= \Os\left(\binom{n/2}{\left\{ \frac{\pi(c)}{2} \; \middle| \; c \in C\right\}}\right) \\
    &= \Os\left(2^{H\left( \left\{ \frac{\pi(c)}{n} \; \middle| \; c \in C\right\} \right) \frac{n}{2}}\right),
\end{align*}
where the second line follows from Stirling's Approximation for multinomial coefficients \eqref{eq:stirling-multinomial}. 

Thus the runtime of Meet-in-the-Middle on unbalanced instances of Subset Balancing depends directly on the entropy of the solution profile: that is,
\begin{equation*}
    H\left( \left\{ \frac{\pi(c)}{n} \; \middle| \; c \in C \right\} \right).
\end{equation*}
This expression evaluates to $\log_2(|C|)$ when $\pi(z) = \frac{n}{|C|}$ for all
$z \in C$, and thus we recover the Meet-in-the-Middle runtime of
$O^*(|C|^{0.5n})$ on perfectly balanced solutions. On the other hand, if $\pi$
is $\epsilon$-unbalanced for any $\epsilon > 0$, the runtime is then
$O(|C|^{(1/2 - \delta)n})$ for some $\delta \coloneqq \delta(\epsilon) > 0$
determined by the entropy function.\footnote{ $\delta \coloneqq \delta(\epsilon)$ is a cumbersome quantity but the closed form can be written down in terms of $\pi$. We omit the expression from this short proof, but incorporate it in the proofs of \cref{prop:SB-pm2-with0} and \cref{prop:SB-pm3-without0} for the specific cases in which $C = [-2: 2]$ and $C = [\pm 3]$. } This concludes the proof of runtime for \Cref{lem:unbalanced}.
\end{proof}

This concludes our discussion of $\epsilon$-unbalanced instances of Subset Balancing. In \Cref{sec:SB-with0} we tackle balanced instances on coefficient sets of the form $C = [-d : d]$.

\subsection{ Balancing Instances Using Re-Randomization }

Although balanced instances of Subset Balancing are the hardest cases for
Meet-in-the-Middle and most other Subset Balancing algorithms, this is not true
for our Equal Subset Sum algorithm in \Cref{sec:ess}. In that section it will be easier (and slightly faster) to deal with instances for which $\pi(1) = \pi(-1)$.
We can assume this holds without loss of generality using the following trick:

\begin{lemma}[Re-randomization trick]\label{lem:rerandomization}
    There exists a randomized, polynomial-time transformation $\tau: \Z^n \rightarrow \Z^n$ defined on instances of Equal Subset Sum with the following properties:
    \begin{itemize}
        \item There is a one-to-one mapping between solutions to $\vec{x}$ and solutions to $\tau(\vec{x})$.
        \item Given a fixed solution vector $\vec{c}$ for $\vec{x}$ with
            solution profile $\pi$, let $\vec{c}\,'$ denote the corresponding
            solution vector for $\tau(\vec{x})$ and let $\pi'$ denote its solution
            profile. Then, $\pi(0) = \pi'(0)$ deterministically and $\pi'(1) = \pi'(-1)$ with probability $n^{-O(1)}$.
    \end{itemize}
\end{lemma}
\begin{proof}
    We define $\tau$ as follows: given a Subset Balancing
    instance $\vec{x} = (x_1, \ldots, x_n)$, independently replace each element
    $x_i$ with $-x_i$ with probability $1/2$. 
    
    Observe that $\tau$ creates a one-to-one mapping between solutions: if $\vec{c}$ is a solution vector satisfying
    $\vec{c} \cdot \vec{x} = 0$, then the vector $\vec{c}\,'$ created by switching
    $c_i$ with $-c_i$ and the vector $\vec{x}\,'$ created by switching $x_i$ with
    $-x_i$ also satisfy $\vec{c}\,' \cdot \vec{x}\,' = 0$. Thus every solution vector
    $\vec{c}$ has a counterpart in the transformed problem. Moreover, any solution
    $\vec{c}\,'$ to the transformed problem can be converted back into a solution for
    the original problem by changing back the sign of the appropriate elements.

    For any fixed solution vector $\vec{c}$ with solution profile $\pi$, the corresponding solution
    vector $\vec{c}\,'$ with solution profile $\pi'$ satisfies $\pi'(0) = \pi(0)$ deterministically as the 0 indices do not change. Moreover, $\pi'(1) = \pi'(-1)$ with probability $n^{-O(1)}$ over
    the choice of indices (by Stirling's approximation for binomial coefficients
    \eqref{eq:stirling}).
\end{proof}

We remark that although we use re-randomization only for $C = [-1 : 1]$, similar lemmas can be proved for all $C = [-d : d]$ and $C = [\pm d]$, $d \geq 1$.

\section{Subset Balancing on \texorpdfstring{$C = [-d : d]$}{[-d:d]} }
\label{sec:SB-with0}

In the previous section, we observed that the Meet-in-the-Middle Barrier can be
broken for Subset Balancing when at least one solution is $\epsilon$-unbalanced
for a constant $\epsilon > 0$. In this section, we design a representation
technique approach for instances of Subset Balancing on coefficient sets of the form $C = [-d : d]$ that have \emph{no} $\epsilon$-unbalanced solutions.

\subsection{ Perfect Mixing and Unbalanced Solutions }
\label{subsec:perf-mixing-dichotomy}

To apply the representation technique to Subset Balancing on $C = [-d : d]$ we define our solution pairs as follows. For a fixed vector $\vec{c} \in C^n$, let
\begin{equation}
   P(\vec{c}) \coloneqq \left\{ (\vec{a}, \vec{b}) \in [0:d]^n \times [0:d]^n \; \middle| \; \vec{a} - \vec{b} = \vec{c} \right\}
   \label{eq:solution-pairs-def}
\end{equation}
denote the set of solution pairs corresponding to $\vec{c}$. We would like
$P(\vec{c})$ to satisfy a \emph{perfect mixing}\footnote{This terminology is
borrowed from Nederlof and Węgrzycki \cite{NederlofW21}, who use an analogous
concept. } property with respect to $\vec{x}$: that is, for $(\vec{a}, \vec{b}),
(\vec{a}\,', \vec{b}\,') \in P(c)$, we would like it to be true that 
\[
    \vec{a} \cdot \vec{x} = \vec{a}\,' \cdot \vec{x} \iff \vec{a} = \vec{a}\,'.
\]
We can guarantee this by proving a mixing dichotomy. We will show that if a Subset Balancing instance $\vec{x}$ admits a solution vector $\vec{c}$ whose solution pairs do \emph{not} mix perfectly with respect to $\vec{x}$, $\vec{x}$ has an unbalanced solution. In either case, we can break the Meet-in-the-Middle barrier.

\begin{lemma}[Perfect Mixing Dichotomy for \texorpdfstring{$C = [-d : d]$}{}]
    Fix an instance $\vec{x}$ of Subset Balancing on a coefficient set $C = [-d : d]$, let $\vec{c}$ be a solution vector, and let $\pi$ be the solution profile of $\vec{c}$.

    Suppose there exist two \emph{distinct} solution pairs $(\vec{a}, \vec{b}),
    (\vec{a}\,', \vec{b}\,') \in P(\vec{c})$ satisfying $\vec{a} \cdot \vec{x} =
    \vec{a}\,' \cdot \vec{x}$. Then $\vec{x}$ admits an $\frac{1}{3|C|}$-unbalanced solution vector.
    \label{lem:perf-mixing}
\end{lemma}
\begin{proof}
    Fix $\vec{x}$, $\vec{c}$, $(\vec{a}, \vec{b})$ and $(\vec{a}\,',
    \vec{b}\,')$ as in the lemma statement. We begin from the fact that
    $P(\vec{c})$ is determined by $\vec{c}$. Specifically, we have that
    $\vec{a}, \vec{b}, \vec{a}\,', \vec{b}\,' \in [0: d]^n$ and that these pairs
    satisfy $\vec{a} - \vec{b} = \vec{a}\,' - \vec{b}\,' = \vec{c}$ by definition. Thus, for each index $i \in [n]$, we have:
    \begin{itemize}
        \item If $c_i = d$, then $a_i = a_i' = d$ and $b_i = b_i' = 0$.
        \item If $c_i = -d$, then $a_i = a_i' = 0$ and $b_i = b_i' = d$.
    \end{itemize}
    In either case, we have that $a_i - a_i' = 0$. Thus the vector $\vec{a} -
    \vec{a}\,'$ contains at least $\pi(d) + \pi(-d)$ zeroes, one for each occurrence of $d$ and $-d$ in $\vec{c}$.

    Next, we observe that $\vec{a} - \vec{a}\,'$ is a solution: because $\vec{a}
    \cdot \vec{x} = \vec{a}\,' \cdot \vec{x}$ by assumption, $(\vec{a} -
    \vec{a}\,') \cdot \vec{x} = 0$. Moreover, because $\vec{a}, \vec{a}\,' \in
    [0:d]^n$, $\vec{a} - \vec{a}\,' \in [-d : d]^n$. $\vec{a} - \vec{a}\,' \neq
    \vec{0}$ as $\vec{a}$ and $\vec{a}\,'$ are distinct by assumption.

    We conclude by showing that either $\vec{c}$ or $\vec{a} - \vec{a}\,'$ must be unbalanced. First, suppose 
    \begin{equation}
        \pi(d) + \pi(-d) < \frac{4}{3} \cdot \frac{n}{|C|}.
        \label{eq:quick-unbalanced}
    \end{equation}
    In this case, $\vec{c}$ is $\frac{1}{3|C|}$-unbalanced as one of $\pi(d)$ or
    $\pi(-d)$ must be less than $\frac{2}{3} \cdot \frac{n}{|C|}$. On the other
    hand, if \eqref{eq:quick-unbalanced} is false, then $\vec{a} - \vec{a}\,'$ is $\frac{1}{3|C|}$-unbalanced, as it has at least $\frac{4}{3} \cdot \frac{n}{|C|}$ zeroes.
\end{proof}

\subsection{ Algorithm for Balanced Subset Balancing on \texorpdfstring{$[-d : d]$}{}.} 
\label{subsec:sb-repmethod-with0}

If a Subset Balancing instance $\vec{x}$ admits a $\frac{1}{3|C|}$-unbalanced solution vector, we can solve the instance quickly using \Cref{lem:unbalanced}. On the other hand, if $\vec{x}$ does not admit a $\frac{1}{3|C|}$-unbalanced solution vector, \cref{lem:perf-mixing} implies that for any solution vector $\vec{c}$, $P(\vec{c})$ mixes perfectly with respect to $\vec{x}$. In this case, we adopt the following representation technique-style algorithm to solve the instance.

\begin{algorithm}[ht!]
    \DontPrintSemicolon
    \nlnonumber\textbf{function} $\textsf{BalancedSBWith0}(\vec{x}, C = [-d : d], \pi)$:\\
    Sample a prime $\bp \sim [p_{max} : 2p_{max}]$ and a residue class $\bm{r} \sim [\bp]$ for $p_{max}$ as defined in \eqref{eq:p-max-balancedSB}.\label{ln:balanced-2}\\
    Enumerate the set\label{ln:balanced-3}
    \begin{align*}
        \calS &\coloneqq \{ \vec{v} \in [0:d]^n \; | \; \vec{v} \cdot \vec{x}
        \equiv \bm{r} \pmod{\bm{p}} \}.
    \end{align*}
    If $|\calS| \geq \frac{(d+1)^n}{p_{max}} \cdot n^{\omega(1)}$, halt and
    return ``\emph{failure}''.\\
    Enumerate and sort the list $L \coloneqq \calS \cdot \vec{x}$.\label{ln:balanced-4}\\
    Use Meet-in-the-Middle to search for a solution pair in $L \times L$.\label{ln:balanced-5}
    \caption{Outline of a representation technique-style algorithm for balanced instances of Subset Balancing on coefficient sets of the form $[-d : d]$. Certain implementation details omitted from this figure are specified in the proof of \Cref{lem:balanced-SB-with0s}.}
    \label{alg:BalancedSBWith0}
\end{algorithm}

\begin{lemma}[Solving Subset Balancing on $\epsilon$-balanced instances with \texorpdfstring{$C = [-d : d]$}{}]
    Let $\epsilon < \frac{1}{3|C|}$ be a constant and let $C = [-d : d]$, $d > 0$ be a coefficient set. There exists an algorithm that solves any Subset Balancing instance $\vec{x}$ on $C$ that does not admit any $\epsilon$-unbalanced solution, with high probability, in time
    \[
        \Os \left( |C|^{(1/2 - \delta)n} \right)
    \] for a constant $\delta > 0$ that depends only on $d$ and $\epsilon$.
    \label{lem:balanced-SB-with0s}
\end{lemma}

Our first step to proving \Cref{lem:balanced-SB-with0s} is an observation that bounds the number of solution pairs for a solution vector $\vec{c}$:

\begin{claim}[Number of Solution Pairs]
    \label{claim:solnpairs}
    Fix a coefficient set $C = [-d : d]$ for an integer $d > 0$ and let $\pi$ be a solution profile. Fix any vector $\vec{c} \in C^n$ and let $P(\vec{c})$ be the set of solution pairs as defined in \eqref{eq:solution-pairs-def}. Then
    \begin{equation*}
       |P(\vec{c})| = \prod_{z \in C} (d + 1 - |z|)^{\pi(z)} = (d+1)^{\pi(0)} \prod_{i \in [d]} (d + 1 - i)^{\pi(i) + \pi(-i)}.
    \end{equation*}
\end{claim}
\begin{claimproof}
    To begin, suppose the $i$th coefficient $s_i$ of $\vec{c}$ is 0. In this case, there are exactly $d + 1$ pairs $(a_i, b_i) \in [0:d] \times [0:d]$ such that $a_i - b_i = 0$: these pairs are $(0, 0), (1, 1), \dots, (d, d)$. 
    There are thus
    \[
        |\{ (\vec{v}, \vec{u}) \in [0:d]^{\pi(0)} \times [0:d]^{\pi(0)} \; \mid \; \vec{v} - \vec{u} = \vec{0} \}| = (d + 1)^{\pi(0)}
    \]
    ways to assign the indices that are $0$ in $\vec{c}$ in a solution pair $(\vec{a}, \vec{b})$.

    Likewise, there are $d$ pairs $(a_i, b_i) \in [0:d] \times [0:d]$ such that $a_i - b_i = 1$: these pairs are $(1, 0), (2,1), \dots, (d, d-1)$. Extending this logic to each coefficient $c \in C$ completes the proof of the claim.
\end{claimproof}

We specify the implementation details omitted from Algorithm
\ref{alg:BalancedSBWith0} in the proof of correctness of
\Cref{lem:balanced-SB-with0s} below. In outline, we set $p_{max}$ to be slightly
less than the number of solution pairs in order that the expected number of
solution pairs in $\calS$ is small but nonzero when solution pairs mix perfectly. 

\begin{proof}[Proof of \Cref{lem:balanced-SB-with0s}: Correctness]
    Fix an instance $\vec{x}$ of Subset Balancing on $C = [-d : d]$, let $\epsilon < \frac{1}{3|C|}$ be a constant, and suppose $\vec{x}$ does not admit any $\epsilon$-unbalanced solution. We assume without loss of generality that $\vec{x}$ admits a solution vector $\vec{c}$ that is $\epsilon$-balanced. Because we can try all $\epsilon$-balanced solution profiles,
    suppose we have correctly guessed the solution profile $\pi$ corresponding to $\vec{c}$.

    Lines~\ref{ln:balanced-4} and \ref{ln:balanced-5} of Algorithm \ref{alg:BalancedSBWith0} will recover a solution deterministically as long as (i) $\calS$ contains a solution pair $(\vec{a}$, $\vec{b})$ satisfying
            $\vec{a} - \vec{b} = \vec{c}$, and (ii) $|\calS| \leq \frac{(d+1)^n}{p_{max}} \cdot n^{O(1)}$.
    Next, we consider the probability that these two events occur when we run Algorithm \ref{alg:BalancedSBWith0} on $\vec{x}$, $C$, and $\pi$. Define the set of \emph{partial solution sums}
    \[
        G \coloneqq \left\{ \vec{a} \cdot \vec{x} \; \middle| \; \exists \vec{b} \text{ s.t. } (\vec{a}, \vec{b}) \in P(\vec{c}) \right\}.
    \]
    with respect to the set of solution pairs $P(\vec{c})$ of $\vec{c}$. By assumption, $\vec{x}$ does not admit any $\epsilon$-unbalanced solution for some $\epsilon < \frac{1}{3|C|}$. Thus, applying the Perfect Mixing Dichotomy (\Cref{lem:perf-mixing}), we have that
    \[
        |G| = |P(\vec{c})|,
    \]
    as every left component of an element of $P(\vec{c})$ has a unique dot product with $\vec{x}$. We set 
    \begin{equation}
        \label{eq:p-max-balancedSB}
    p_{max} = \min \left\{ |G|, (d+1)^{n/2} \right\}.
    \end{equation}
    Here, the first argument to the $\min$ function reflects the case in which we intend that many residue classes modulo $\bp$ contain $\Theta^*(1)$ elements of $G$. The second argument caps $p_{max}$. This is done to minimize the overall runtime of the algorithm, which would otherwise scale with $|G|$ in cases where $|G|$ is very large.

    We proceed to show that the algorithm recovers a solution with constant
    probability, by analysing two cases that depend on the value of $p_{max}$ assigned
    in~\eqref{eq:p-max-balancedSB}.
    \subparagraph*{Case 1: $p_{max} = |G|$.} Apply \Cref{lem:prime-dist-main}
    using the sets $G$ and $([0:d]^n \cdot \vec{x})$ and the prime range bound
    $p_{max}$. Let $b$ be a constant satisfying $\diam(G) < 2^{n^{b}}$, and note
    that $b$ \emph{is} constant under our assumption that input integers have
    size $2^{n^{O(1)}}$. As a result, we have that with constant probability the
    set $\calS$ (i) contains at least $\frac{|G|}{4p_{max}}  = 1/4 > 0$ elements
    of $G$, and (ii) satisfies $|\calS| \leq \frac{(d+1)^n}{p_{max}} \cdot n^{b+1}$.
    If the set $\calS$ satisfies both conditions, Algorithm \ref{alg:BalancedSBWith0} recovers a solution deterministically.

    \subparagraph{Case 2: $p_{max} = (d+1)^{n/2}$.}
    Once again, apply \Cref{lem:prime-dist-main} using the sets $G$ and
    $([0:d]^n \cdot \vec{x})$ and the prime range bound $p_{max}$; define
    $b$ as before. We have that with constant probability the set $\calS$
    (i) contains $\frac{|G|}{4p_{max}} > 1/2 > 0$ elements of $G$, and (ii)
    satisfies $|\calS| \leq \frac{(d+1)^n}{p_{max}} \cdot n^{b+1} =
    (d+1)^{n/2} \cdot n^{b+1}$.

    Once again, if the set $\calS$ satisfies both conditions, Algorithm \ref{alg:BalancedSBWith0} recovers a solution deterministically.\footnote{Note that in this case, the number of good elements $\frac{|G|}{4p_{max}}$ may be very large, in which case we could improve our runtime significantly using a subsampling approach. However, this is not necessary to beat Meet-in-the-Middle (and for $C = [-d:d]$, Case 2 does not apply for small values of $d$), so we leave this as an exercise.}
    Thus Algorithm \ref{alg:BalancedSBWith0} recovers a solution with constant probability if we correctly guess the solution profile $\pi$ corresponding to $\vec{c}$. Amplifying the success probability completes the proof of correctness for \Cref{lem:balanced-SB-with0s}.
\end{proof}

\begin{proof}[Proof of \Cref{lem:balanced-SB-with0s}: Runtime]
    It remains to show that, given a Subset Balancing instance $\vec{x}$, coefficient set $C$, and $\epsilon$-balanced solution profile $\pi$, Algorithm \ref{alg:BalancedSBWith0} runs in time $O(|C|^{(1/2 - \delta)n})$ for some constant $\delta > 0$ that depends only on $d$ and $\epsilon$.

    Line~\ref{ln:balanced-2} can be performed in polynomial time with high probability by random sampling and polynomial-time primality-checking, as $p_{max} = 2^{n^{O(1)}}$ under our assumption that input integers have size $2^{n^{O(1)}}$.

    Lines~\ref{ln:balanced-4} and~\ref{ln:balanced-5} can be easily implemented in time
    \[
        \Os(|\calS|) = \Os\left( \frac{(d+1)^n}{p_{max}} \right).
    \]

    It remains to describe the implementation of Line~\ref{ln:balanced-3} and bound the runtime. We implement Line~\ref{ln:balanced-3} as follows.
    First, we allocate a dynamic programming table with dimensions $n \times
    \bp$. Then, for $i \in [n]$, $j \in [\bp]$, store the number of vectors $\vec{a} \in [0:d]^i$ satisfying 
        \[
			\vec{a}[0:i-1] \cdot \vec{x}[0:i-1] \equiv j \pmod{\bm{p}}
        \]
		in cell $(i, j)$; that is, the number of ``partial solution vectors'' that allow us to achieve the sum $j \pmod{\bm{p}}$.
        Each cell $(i, j)$ can be filled in time $O(|C|)$ by computing the sum of the entries of $|C|$ cells in the previous row. Thus filling in the entire table takes time $O^*(p_{max})$.
        Finally, we enumerate $\calS$ in time $\Os(|\calS|)$ by backtracking from
        cell $[n-1, \mathbf{r}]$. (If $|\calS| = \omega\left(\frac{(d+1)^n}{p_{max}}\right)$, the algorithm halts after taking $O^*\left(\frac{(d+1)^n}{p_{max}}\right)$ steps.)

    Taking the maximum over Lines~\ref{ln:balanced-3}-\ref{ln:balanced-5}, Algorithm \ref{alg:BalancedSBWith0} takes time
    \[
    \Os \left( \max \left\{ p_{max}, \frac{(d+1)^n}{p_{max}} \right\} \right) = \Os \left( \frac{(d+1)^n}{p_{max}}\right),
    \]
    as $p_{max} \leq (d+1)^{n/2}$ by definition.
    It remains to show that this is quantity is less than the Meet-in-the-Middle runtime of $O^*(|C|^{n/2}) = O^*((2d+1)^{n/2})$ for $C = [-d : d]$, $d \geq 1$.

    The running time depends on $p_{max}$ which based on \eqref{eq:p-max-balancedSB} can take two values.
    When $p_{max} = (d+1)^{n/2}$, the running time is 
    $\Os \left( (d+1)^n/p_{max}\right) = \Os((d+1)^{n/2})$,
    which is less than $(2d+1)^{n/2}$ by an exponential margin for any positive integer $d$.
    Therefore, we need to analyse the case when $p_{max} = |G|$.
    Substituting for $|G|$ using \Cref{claim:solnpairs}, in this case we have a runtime of
            \[
            \Os \left( \frac{(d+1)^n}{|G|}\right) = \Os \left( \frac{(d+1)^n}{(d+1)^{\pi(0)} \prod_{i \in [d]} (d+1-i)^{\pi(i) + \pi(-i)}}\right).
            \]
        By assumption, our solution vector $\vec{c}$ is not $\frac{1}{3|C|}$-unbalanced, so $\pi(i) \in \frac{n}{|C|} \pm \frac{n}{3|C|}$ for all $i \in C$. Plugging in the lower bound $\pi(i) \geq \frac{2}{3|C|}n$ for each $i \in C$,\footnote{Note that this effectively assumes the worst case among all $(\frac{1}{3|C|})$-balanced solution profiles. This is sufficient to break Meet-in-the-Middle, but we can optimize the runtime by calculating it as a function of $\pi$ for both Algorithm~\ref{alg:BalancedSBWith0} and Algorithm~\ref{alg:UnbalancedSB} and choosing whichever is faster for each possible $\pi$. We show an example of this in \Cref{subsec:runtime-opt-pm2}.} we have a runtime upper bound of 
            \[
            \Os \left( \frac{(d+1)^n}{(d+1)^{\frac{2n}{3(2d+1)}} (d!)^{\frac{4n}{3(2d+1)}}} \right).
            \]

            In~\cref{claim:C0-runtime-analysis} we show that 
        this quantity is exponentially smaller than the runtime of
        Meet-in-the-Middle for every positive integer $d$. The
proof is straightforward but somewhat technical hence we defer it to \Cref{apx:C0-runtime-lemma-analysis}.
        
\end{proof}

\subsection{ Breaking Meet-in-the-Middle for Subset Balancing on \texorpdfstring{$C = [-d : d]$}{} }

\begin{theorem}
    \label{thm:main-with0}
    For any $d > 0$, there exists a constant $\delta > 0$ such that Subset Balancing on $C = [-d : d]$ can be solved with high probability in time 
    \[
        \Os \left( |C|^{(1/2 - \delta)n}\right).
    \]
\end{theorem}
\begin{proof}
    Choose a constant $\epsilon < \frac{1}{3|C|}$ and apply \Cref{lem:unbalanced,lem:balanced-SB-with0s}.
\end{proof}

The optimal constant $\delta$ in the exponent results from trading off the approaches for balanced and unbalanced solutions to find the worst-case solution profile $\pi$. This becomes a messy exercise for large $d$. As a proof of concept, we solve for $\delta$ for $C = \{-2, -1, 0, 1, 2\}$ in the next subsection. 

\subsection{ Optimizing Runtime for \texorpdfstring{$C = [-2 : 2]$}{} }
\label{subsec:runtime-opt-pm2}

In this section, we optimize $\delta$ in the statement of \Cref{thm:main-with0} for $C = [-2 : 2]$ as a proof of concept. Specifically, we prove the following proposition.

\begin{proposition}
    \label{prop:SB-pm2-with0}
    Subset Balancing on $C = [-2 : 2]$ can be solved with high probability in time
    \[
        \Os(|C|^{0.478n}) = \Os(2^{1.108n}).
    \]
\end{proposition}
\begin{proof}
    Optimizing runtime comes down to trading off the runtimes given by Algorithm \ref{alg:UnbalancedSB} and Algorithm \ref{alg:BalancedSBWith0}. Specifically, we will calculate a closed-form expression for the time required to find a solution in terms of the solution profile $\pi$ of a fixed solution vector $\vec{c}$, then take the maximum over all possible $\pi$.

    Consider an instance $\vec{x}$ of Subset Balancing on $C = [-2 : 2]$, and suppose $\vec{x}$ admits a solution vector $\vec{c}$ with solution profile $\pi$.

    \begin{itemize}
        \item \textbf{Case 1: Unbalanced Solutions.} Following the proof of runtime for \Cref{lem:unbalanced}, we have that the runtime of Algorithm \ref{alg:UnbalancedSB} given $\pi$ is
        \[
            \Os\left(2^{H\left(\frac{\pi(-2)}{n}, \frac{\pi(-1)}{n}, \frac{\pi(0)}{n}, \frac{\pi(1)}{n}, \frac{\pi(2)}{n}\right) \frac{n}{2}} \right)
        \]
        \item \textbf{Case 2: Balanced Solutions.} Following the proof of runtime for \Cref{lem:balanced-SB-with0s}, we have that the runtime of Algorithm \ref{alg:BalancedSBWith0} given $\pi$ is
        \[
            \Os \left( \frac{(d+1)^n}{p_{max}} \right) 
            = \Os \left( \max\left[\frac{3^n}{3^{\pi(0)} \cdot 2^{\pi(-1) + \pi(1)}}, 3^{n/2} \right]\right).
        \]
        The $3^{n/2}$ term is not binding for this coefficient set, so we ignore it in the following optimization.
    \end{itemize}

    At this point, we can simplify the task of finding the worst-case solution profile $\pi$ by making several observations:
    \begin{enumerate}
        \item For any fixed values of the quantities $(\pi(-2) + \pi(2))$ and $(\pi(-1) + \pi(1))$, the Case 1 runtime is maximized by choosing to set $\pi(-2) = \pi(2)$ and $\pi(-1) = \pi(1)$, while the Case 2 runtime is unchanged. Thus we can proceed under this assumption without loss of generality and write the expressions above in terms of $\pi(0)$, $\pi(1)$ and $\pi(2)$.
        \item The solution profile is constrained by the identity $\sum_{c \in [-2 : 2]} \pi(c) = n$. Thus we can reduce the number of variables further by writing $\pi(2)$ in terms of $\pi(0)$ and $\pi(1)$.
        \item Both runtimes can easily be written as functions of the form $2^{O(n)}$, so any variable assignment that balance the exponents of these functions also balances the original functions.
    \end{enumerate}

    Thus, writing both functions in terms of $\alpha_0 = \pi(0)/n$ and $\alpha_1 = \pi(1)/n$ and taking the base-2 logarithm of both sides, we can balance our runtime for worst-case $\pi$ by solving the following optimization problem:

    \begin{equation}
        \max_{\alpha_0, \alpha_1 \in [0:1]} \min \left\{ \frac{1}{2} \cdot
            H\left(\alpha_0, \alpha_1, \alpha_1, \frac{1 - \alpha_0 - 2\alpha_1}{2}, \frac{1 - \alpha_0 - 2\alpha_1}{2}\right), \;
            \log_2(3) (1 - \alpha_0) - 2\alpha_1
    \right\}
        \label{eq:2-2-to-opt}
    \end{equation}

    Computer evaluation reveals that \eqref{eq:2-2-to-opt} is upper-bounded by $1.108$ and that the maximum occurs at approximately $(\alpha_0, \alpha_1) = (0.105, 0.156)$. 
\end{proof}

\section{ Subset Balancing on \texorpdfstring{$[\pm d]$}{} }
\label{sec:SB-without0}

In \Cref{sec:SB-with0}, we designed faster worst-case algorithms for Subset Balancing on coefficient sets of the form $C = [-d : d]$, which generalize Equal Subset Sum. In this section, we consider worst-case algorithms for Subset Balancing  on coefficient sets of the form $C = [\pm d]$, which generalize Subset Sum.

Unlike Equal Subset Sum, no $\Oh(|C|^{(0.5 - \epsilon)n})$-time algorithm for
Subset Sum is known for any constant $\epsilon > 0$. We proceed to show that,
perhaps surprisingly, $\Oh(|C|^{(0.5 - \epsilon)n})$-time algorithms \emph{are} possible on $[\pm d]$ for any $d > 2$.

\subsection{ The Coefficient Shifting Approach }
\label{subsec:coeff-shifting}

Our algorithm for balanced instances of Subset Balancing on coefficient sets of the form $C = [-d : d]$ crucially relies on the fact that for \emph{any} two vectors $\vec{a}, \vec{b} \in [0:d]^n$, $\vec{a} - \vec{b} \in C^n$. This property makes it relatively easy to represent any solution vector by a set of solution pairs. 

However, this property does not hold for $C = [\pm d]$ because the $0$
coefficient is missing. For the sake of intuition, consider what would happen if
we attempted to apply Algorithm \ref{alg:BalancedSBWith0} to an instance of
Subset Balancing on a coefficient set of the form $[\pm d]$. Our perfect mixing
dichotomy (\Cref{lem:perf-mixing}) does not apply: when two partial solutions
$(\vec{a}, \vec{b})$ and $(\vec{a}\,', \vec{b}\,')$ collide in the sense that
$\vec{a} \cdot \vec{x} = \vec{a}\,' \cdot \vec{x}$, we cannot conclude that
$(\vec{a} - \vec{a}\,')$ is a solution vector because $(\vec{a} - \vec{a}\,')$ might be $0$ at some indices.

However, it is sometimes possible to avoid this issue by using coefficient sets other than $[d]$ or $[0:d]$ to represent partial candidate solutions, a new approach we call \emph{coefficient shifting}. For a concrete example, consider the coefficient set 

\[
    C \coloneqq [\pm 3] = \{-3, -2, -1, 1, 2, 3\}.
\]

The key insight is that $C$ is the sumset of many pairs of integer sets. For
example, for $C_1 \coloneqq \{0, 1\}$ and $C_2 \coloneqq \{-3, -2, 1, 2\}$, the equation
\begin{equation}
    \label{eq:coeff-decomposition}
    C_1 + C_2 = C
\end{equation}
holds. We will refer to sets $C_1$, $C_2$ with this property as a pair of \emph{coefficient set factors} of $C$.

Suppose we pursue an unusual variant of the Meet-in-the-Middle strategy by defining the following two lists: 
\begin{align}
    \label{eq:L-R-veclists-def}
    \mathcal{L} \coloneqq \mathcal{L}(C_1, C_2) &= C_1^{n/2} \times C_2^{n/2} \text{ and } \\
    \mathcal{R} \coloneqq \mathcal{R}(C_1, C_2) &= C_2^{n/2} \times
    C_1^{n/2}\text{.}\nonumber
\end{align}
This construction ensures that both lists have the same length: 
\[
    |\mathcal{L}| = |\mathcal{R}| = 2^{n/2} \cdot 4^{n/2} = 2^{3n/2}.
\]
Moreover, because $C_1$ and $C_2$ are coefficient set factors of $C$, every vector $\vec{c} \in C^n$ is represented by a non-empty set of \emph{solution pairs}: 
\begin{equation}
    \label{eq:solution-pairs-shifted-def}
    P(\vec{c}, C_1, C_2) = \left\{ (\vec{a}, \vec{b}) \in \mathcal{L} \times \mathcal{R} \; \middle| \; \vec{a} + \vec{b} = \vec{c} \right\}.\footnote{In this section, we add the two solution pair elements instead of subtracting as in \Cref{sec:SB-with0} for notational convenience. The two conventions are mathematically equivalent.}
\end{equation}
Given an input vector $\vec{x} \in \Z^n$, we can solve the instance by computing
two lists $L \coloneqq \vec{x} \cdot \mathcal{L}$ and $R \coloneqq \vec{x} \cdot \mathcal{R}$, sorting, and using the Meet-in-the-Middle algorithm to recover a solution pair in time $O^*(2^{3n/2})$ if one exists.

For our chosen coefficient set $C$ this runtime is slower than Meet-in-the-Middle, which runs
in time $O^*(|C|^{n/2}) = O^*(6^{n/2})$ for $C = [\pm 3]$. However, we can
improve the runtime further by applying the representation technique to take advantage of the many solution pairs corresponding to each solution vector. In order to do so, we would like to show something akin to the following two conditions:
\begin{enumerate}
    \item Our instance $\vec{x}$ of Subset Balancing admits an $\epsilon$-balanced solution $\vec{c}$ for some small constant $\epsilon > 0$. This is easy to ensure without loss of generality, as if the instance admits an $\epsilon$-unbalanced solution, we can solve it faster than Meet-in-the-Middle by \Cref{lem:unbalanced}.
    \item The set of solution pairs of $\vec{c}$
    satisfies a perfect mixing property with respect to $\vec{x}$: for solution
    pairs $(\vec{a}, \vec{b})$, $(\vec{a}\,', \vec{b}\,')$, we have $\vec{a}
    \cdot \vec{x} = \vec{a}\,' \cdot \vec{x}$ if and only if $\vec{a} =
    \vec{a}\,'$.

    This is harder to achieve because, as explained above, our original perfect mixing dichotomy does not apply. However, we can replace it with a different mixing dichotomy: we will show that either the solution pairs of $\vec{c}$ mix (almost) perfectly or $\vec{x}$ admits \emph{exponentially many} solution vectors, in which case we can break the Meet-in-the-Middle barrier using a simple random sampling approach.
\end{enumerate}

Finally, we will apply the representation technique to filter the lists $\mathcal{L}$ and $\mathcal{R}$ and achieve worst-case runtime improvements for all coefficient sets $C = [\pm d]$, $d \geq 2$. As before, we content ourselves with showing that we can achieve $O(|C|^{(1/2 - \delta)n})$ for some constant $\delta > 0$ that depends only on $|C|$ in the general case. In \Cref{subsec:runtime-opt-pm3} we calculate $\delta$ exactly for the example coefficient set $C = [\pm 3]$.

\subsection{ A Mixing Dichotomy for Subset Balancing on \texorpdfstring{$[\pm d]$}{} }
\label{subsec:sb-without0-mixing}

Given a coefficient set $C$, coefficient set factors $C_1$ and $C_2$, and a coefficient vector $\vec{c} \in C^n$, define the two lists $\mathcal{L}$ and $\mathcal{R}$ as well as the set of solution pairs $P(\vec{c}, C_1, C_2)$ as in \eqref{eq:L-R-veclists-def} and \eqref{eq:solution-pairs-shifted-def} above. In addition, we define two sets containing the left and right components of the elements of $P(\vec{c}, C_1, C_2)$, respectively:
\begin{align*}
    P_{\mathcal{L}} \coloneqq P_{\mathcal{L}}(\vec{c}, C_1, C_2) &= \{ \vec{a} \in (C_1^{n/2} \times C_2^{n/2})   \; | \; \exists \vec{b} \in (C_2^{n/2} \times C_1^{n/2}), \vec{a} + \vec{b} = \vec{c}\} \\
    P_{\mathcal{R}} \coloneqq P_{\mathcal{R}}(\vec{c}, C_1, C_2) &= \{ \vec{b} \in (C_2^{n/2} \times C_1^{n/2})   \; | \; \exists \vec{a} \in (C_1^{n/2} \times C_2^{n/2}), \vec{a} + \vec{b} = \vec{c}\}.
\end{align*}
We are now prepared to prove our new mixing dichotomy.

\begin{lemma}[Mixing Dichotomy for \texorpdfstring{$[\pm d]$}{}]
    Fix a coefficient set $C$, coefficient set factors $C_1$ and $C_2$, and an instance $\vec{x}$ of Subset Balancing. Let $\vec{c} \in C^n$ be a solution vector for $\vec{x}$.

    For any constant $\epsilon > 0$, if
    \[
        |\{P_{\mathcal{R}} \cdot \vec{x}\}| \leq |P(\vec{c}, C_1, C_2)| \cdot 2^{-\epsilon n}
    \]
    then $\vec{x}$ has at least $2^{\epsilon n}$ distinct solution vectors.
    \label{lem:perf-mixing2}
\end{lemma}
\begin{proof}
    Fix $C$, $C_1$, $C_2$, $\vec{c}$, and $\vec{x}$ as above, and suppose that 
    \begin{equation}
        |\{P_{\mathcal{R}} \cdot \vec{x}\}| \leq |P(\vec{c}, C_1, C_2)| \cdot 2^{-\epsilon n}
        \label{eq:SB-shifting-collisions}
    \end{equation}
    for some constant $\epsilon > 0$. First, we observe that 
    \[
        |P_{\mathcal{L}}| = |P_{\mathcal{R}}| = |P(\vec{c}, C_1, C_2)|.
    \]
    This is because there exist one-to-one mappings between each element $\vec{\ell} \in P_{\mathcal{L}}$, its partner $\vec{c} - \vec{\ell} \in P_{\mathcal{R}}$, and the solution pair $(\vec{\ell}, \vec{c} - \vec{\ell}) \in P(\vec{c}, C_1, C_2)$.

    Thus \eqref{eq:SB-shifting-collisions} implies the existence of exponentially many collisions when we take the dot product of $P_{\mathcal{L}}$ and $\vec{x}$. In particular, there exists some integer $v$ such that the set of left solution pairs colliding at $v$
    \[
        P_{\mathcal{L}, v} \coloneqq \left|\left\{ \vec{\ell} \in P_{\mathcal{L}} \; | \; \vec{\ell} \cdot \vec{x} = v \right\}\right| \geq 2^{\epsilon n}.
    \]
    Fix an arbitrary element $\vec{\ell} \in P_{\mathcal{L}, v}$ and consider its solution pair $(\vec{\ell}, \vec{c} - \vec{\ell})$.\

    We claim that the set 
    \[
        Q \coloneqq (\vec{c} - \vec{\ell}) + P_{\mathcal{L}, v}
    \]
    is a set of $2^{\epsilon n}$ distinct solution vectors for $\vec{x}$. To see this, first observe that $|Q| = |P_{\mathcal{L}, v}| \geq 2^{\epsilon n}$, as $Q$ is a translation of $P_{\mathcal{L}, v}$. Second, we observe that every vector $(\vec{c} - \vec{\ell} + \vec{q}) \in Q$ is a solution vector:
    \begin{align*}
        ((\vec{c} - \vec{\ell}) + \vec{q}) \cdot \vec{x} &= (\vec{c} - \vec{\ell}) \cdot \vec{x} + \vec{q} \cdot \vec{x} \\
        &= (\vec{c} - \vec{\ell}) \cdot \vec{x} + v \\
        &= (\vec{c} - \vec{\ell}) \cdot \vec{x} + \vec{\ell} \cdot \vec{x} \\
        &=  \vec{c} \cdot \vec{x} = 0,
    \end{align*}
    where the second and third steps follow from the definition of $P_{\mathcal{L}, v}$. Finally, as $\vec{c} - \vec{\ell} \in P_{\mathcal{R}} \subseteq C_2^{n/2} \times C_1^{n/2}$ and $\vec{q} \in P_{\mathcal{L}} \subseteq C_1^{n/2} \times C_2^{n/2}$, we know that $(\vec{c} - \vec{\ell} + \vec{q}) \in C^n$ because $C_1$ and $C_2$ are coefficient set factors of $C$. This completes the proof of \Cref{lem:perf-mixing2}.
\end{proof}

The property that $\vec{x}$ has many distinct solution vectors is useful because such instances can be solved by randomly sampling partial solutions before running Meet-in-the-Middle:
\begin{lemma}
    \label{lem:SB-many-solutions}
    Let $\vec{x}$ be an instance of Subset Balancing on a coefficient set $C$. If there exists a constant $\epsilon > 0$ such that $\vec{x}$ admits at least $2^{\epsilon n}$ distinct solution vectors $\vec{c} \in C^n$, $\vec{x}$ can be solved with high probability in time
    \[
        \Os(|C|^{n/2} \cdot 2^{-\epsilon n / 2})
    \]
\end{lemma}
\begin{proof}
    Create two lists $L$ and $R$ by sampling 
    $|C|^{n/2} \cdot 2^{-\epsilon n / 2}$
    elements of $C^{n/2}$ uniformly at random with replacement. Compute $L \cdot \vec{x}$ and $R \cdot \vec{x}$ and sort them by sum in time $\Os(|C|^{\frac{1-\epsilon}{2} n})$. Finally, run Meet-in-the-Middle on $L \times R$, which takes time $O(|L| + |R|)$.

    Meet-in-the-Middle recovers a solution deterministically if a solution is sampled; that is, if there exists $(\vec{a}, \vec{b}) \in L \times R$ such that $\vec{a} \circ \vec{b}$ is a solution vector for $\vec{x}$, where $\circ$ concatenates components. Standard concentration inequalities imply that this event occurs with constant probability; this can be amplified to high probability by repeating the algorithm $\poly(n)$ times.
\end{proof}

\subsection{ Choosing New Coefficient Sets }
\label{subsec:num-soln-pairs-without0}

In this section, we address the following question: which sets $C$ have coefficient set factors $C_1$ and $C_2$ that make the coefficient shifting approach possible? 

Roughly speaking, we want the length of the lists $\mathcal{L}$ and $\mathcal{R}$ to be less than $|C|^{n/2}$ when we hash them modulo our large random prime $\bp$, which will be approximately the size of the set $|P(\vec{c},C_1,C_2)|$ for a perfectly balanced solution vector $\vec{c}$. (Here we choose a balanced solution vector because we can easily solve unbalanced instances using \Cref{lem:unbalanced} and thus the ``hard'' instances of Subset Balancing will have solutions approximating this form.) This condition turns out to be sufficient.

\begin{definition}[Good coefficient set factors]
    \label{defn:good-factors}
    For a fixed coefficient set $C$, integer sets $C_1$ and $C_2$ are a pair of \emph{good coefficient set factors} 
    if they satisfy the following conditions for any constant $\epsilon > 0$:
    \begin{align*}
        C_1 + C_2 &= C \\
        \frac{|\mathcal{L}(C_1, C_2)|}{|P(\vec{c},C_1,C_2)|} &= o(|C|^{n/2} \cdot 2^{-\epsilon n}),
    \end{align*}
    where $\vec{c}$ is any \emph{perfectly balanced} coefficient vector.
\end{definition}
As it turns out, all coefficient sets of the form $C = [\pm d]$, $d > 2$ admit good coefficient set factors.

\begin{lemma}
    \label{lem:good-coeff-sets}
    Every set of the form $C = [\pm d]$, $d > 2$ has a pair of good coefficient set factors.
\end{lemma}
\begin{proof}
    Fix a coefficient set $C = [\pm d]$, $d > 2$, and consider the coefficient sets
    \begin{align*}
        C_1 &\coloneqq \{0, 1\} \text{ and } \\
        C_2 &\coloneqq \{-d, -d+1, \dots, -2, 1, 2, \dots, d-1\}.
    \end{align*}
    Note that $C_1 + C_2 = C$ and $|\mathcal{L}(C_1, C_2)| = 2^{n/2} \cdot (|C|-2)^{n/2}$. We fix a perfectly balanced coefficient vector $\vec{c}$ and proceed to count $|P(\vec{c},C_1,C_2)|$.\footnote{Note that $P(\vec{c},C_1,C_2)$ depends only on $\vec{c}$'s solution profile, so the specific choice of $\vec{c}$ does not matter as long as $\vec{c}$ is balanced.}

	We first observe that every coefficient in $C = [\pm d]$, with the four exceptions of $-d$, $-1$, $1$ and $d$, can be represented in exactly two ways as a sum of elements in $C_1 \times C_2$.
	Namely, any $i \in \{-d+1,\ldots,-2\} \cup \{2,\ldots,d-1\}$ can be represented as $i = 0 + i$ or $k = 1 + (i-1)$. Thus
    \begin{equation}
        |P(\vec{c},C_1,C_2)| = 2^{\frac{n}{|C|} \cdot (|C| - 4)} = 2^{\frac{|C|-4}{|C|} n}.
        \label{eq:P-without0-goodfactors-count}
    \end{equation}
    It remains to show that
    \[
        \frac{|\mathcal{L}(C_1, C_2)|}{|P(\vec{c},C_1,C_2)|} = 2^{n/2} \cdot (|C|-2)^{n/2} \cdot 2^{-\frac{|C|-4}{|C|} n} \leq |C|^{n/2} \cdot 2^{-\epsilon n}
    \]
    holds for some constant $\epsilon > 0$. For $d = 3$ and $d = 4$ one can
    individually check that this holds for $\epsilon > 0$ (see~\cref{table:coeff-shifting}).
    Hence, we focus on $d \ge 5$.  Manipulating the left side, we have
    \begin{align*}
        2^{n/2} \cdot (|C|-2)^{n/2} \cdot 2^{-\frac{|C|-4}{|C|} n} &< 
        2^{n/2} \cdot |C|^{n/2} \cdot 2^{-\frac{|C|-4}{|C|} n} \\
        &= |C|^{n/2} \cdot 2^{\left(\frac{4}{|C|}-\frac{1}{2}\right)n} \\
        &= |C|^{n/2} \cdot 2^{-\epsilon n}
    \end{align*}
    for some $\epsilon \geq \frac{1}{10}$, as $|C| = |[\pm d]| \geq 10$ for $d
    \ge 5$.
\end{proof}

At this point, we pause to make one technical observation that is convenient (although not strictly necessary) for the proof of \Cref{lem:balancedSB-without0} below.
\begin{observation}
    \label{obs:p-small-without0}
    For every set of the form $C = [\pm d]$, $d > 2$, there exists a pair of good coefficient set factors $C_1$, $C_2$ satisfying
    \[
        |P(\vec{c},C_1,C_2)| < 2^n
    \]
    for some $\epsilon > 0$.
\end{observation}
\begin{proof}
    This property holds for the choice of good coefficient set factors used in the proof of \Cref{lem:good-coeff-sets}, by \eqref{eq:P-without0-goodfactors-count}.
\end{proof}

To break Meet-in-the-Middle for a coefficient $C$, we will need only the fact
that $C$ has a pair of good coefficient set factors, but the magnitude of the
improvement depends on exactly ``how good'' the pair $C_1$, $C_2$ is. The
specific choice of $C_1$ and $C_2$ used in the proof of
\Cref{lem:good-coeff-sets} simplifies the argument, but in general, better
choices of coefficient set factors are possible. In fact, if we choose the best
coefficient set factors the runtime improvement relative to
Meet-in-the-Middle appears to \emph{increase} with $|C|$. \Cref{table:coeff-shifting}
summarizes the possible choices for good coefficient set factors for $C = [\pm d]$, $3 \leq d \leq 7$.

\newcommand{\TableCell}[1]{
\begin{tabular}{l}
    #1
  \end{tabular}
}
\renewcommand{\arraystretch}{1.3}
 \begin{table}[]
    \resizebox{\textwidth}{!}{%
        \begin{tabular}{|c|l|l|l|l|l|}
        \hline
        \textbf{$C$} & \textbf{$C_1$}, \textbf{$C_2$} & \textbf{$|\mathcal{L}(C_1, C_2)|$} & \textbf{$|P(\vec{c},C_1,C_2)|$} & $\dfrac{|\mathcal{L}|}{|P(\vec{c},C_1,C_2)|} $& $|C|^{n/2}$ \\ \hline
        $[\pm 3]$  &\TableCell{$\{0, 1\},\{-3, -2, 1, 2\}$} &  $8^{n/2}$ &
        $(2^{n/6})^2$ & $\le 2.245^n$ & $6^{n/2} \ge 2.449^n$  \\ \hline
          $[\pm 4]$  &
          \TableCell{
    $\{0, 1\},\{-4, -3, -2, 1, 2, 3\}$ \\
    $\{0,1,2\}, \{-4, -3, 1, 2\}$
} & $12^{n/2}$ & $(2^{n/8})^4$ & $\le 2.450^n$   &$8^{n/2} \ge 2.828^n$\\ \hline
$[\pm 5]$  &\TableCell{$\{0, 1\},\{-5,-4,-3,-2,1,2,3,4\}$\\$\{0, 1, 2,
          3\}, \{-5, -4, 1, 2\}$}&$16^{n/2}$&$(2^{n/10})^6$&$\le 2.640^n$&$10^{n/2}
          \ge 3.162^n$ \\\cline{2-5}
          &\TableCell{$\{0,1,2\},\{-5,-4,-3,1,2,3\}$}&$18^{n/2}$&$(2^{n/10})^4(3^{n/10})^2$&$\le 2.582^n$&\\ \hline
          $[\pm 6]$  &\TableCell{$\{0, 1\},\{-6, \dotsc, -2,1,\dotsc,5\}$\\$\{0, 1, 2,
          3,4\}, \{-6, -5, 1, 2\}$} &$20^{n/2}$&$(2^{n/12})^8$
          &$\le 2.818^n$&$12^{n/2} \geq 3.464^n$\\\cline{2-5}
          &\TableCell{$\{0, 1,2\},\{-6, \dotsc, -3,1,\dotsc,4\}$\\$\{0,
          1, 2,3\}, \{-6,-5,-4, 1,2,3\}$}&$24^{n/2}$&$(2^{n/12})^4(3^{n/12})^4$&$\le 2.697^n$&\\ \hline
          &\TableCell{$\{0, 1\},\{-7,
              \dotsc,-2,1,\dotsc,6\}$\\$\{0,1,2,3,4,5\},
      \{-7,-6,1,2\}$}&$24^{n/2}$&$(2^{n/14})^{10}$&$\le 2.987^n$&\\\cline{2-5}
      $[\pm 7]$  &\TableCell{$\{0, 1,2\},\{-7,
          \dotsc,-3,1,\dotsc,5\}$\\$\{0,1,2,3,4\},
      \{-7,-6,-5,1,2,3\}$}&$30^{n/2}$&$(2^{n/14})^4(3^{n/14})^6$&$\le 2.807^n$&$14^{n/2}
      \ge 3.741^n$\\\cline{2-5}
      &\TableCell{$\{0, 1,2,3\},\{-7,\dotsc,-4,1,2,3,4\}$}&$32^{n/2}$&$(2^{n/14})^4(3^{n/14})^4(4^{n/14})^2$&$\le 2.782^n$&\\ \hline
        \end{tabular}
    }
    \caption{Good coefficient set factors for $C = [\pm d]$, $3 \leq d \leq 7$. The runtimes in the fourth column approximate the runtime of our algorithm on perfectly balanced solution profiles. (The worst-case running time occurs on slightly unbalanced solution profiles and results from taking the better of this approach and the previous approach for unbalanced solution profiles.)}
    \label{table:coeff-shifting}
    \end{table}

\subsection{ Algorithm for Balanced Subset Balancing on \texorpdfstring{$[\pm d], d > 2$}{} }
\label{subsec:sb-without0-repmethod}

\begin{algorithm}[ht!]
    \DontPrintSemicolon
    \nlnonumber\textbf{function} $\textsf{BalancedSBWithout0}(\vec{x}, C = [\pm d]), \pi)$:\\
    Choose good coefficient set factors $C_1$ and $C_2$ for $C$.\label{ln:nozero-1} \tcp*{see~\Cref{defn:good-factors}}
    Set $\gamma \coloneqq \gamma(C_1, C_2)$.\label{ln:nozero-2}\tcp*{see \eqref{eq:gamma-without0}}
    Solve if $\vec{x}$ admits $2^{\frac{\gamma}{2}n}$ distinct solution vectors.\label{ln:nozero-3}\tcp*{see~\Cref{lem:SB-many-solutions}}
    Set $p_{max} \coloneqq p_{max}(C_1, C_2)$.\label{ln:nozero-4} \tcp*{see \eqref{eq:pmax-without-0}}
    Sample a prime $\bp \sim [p_{max} : 2p_{max}]$ and a residue class $\bm{r} \sim \bp$.\label{ln:nozero-5}\\
    Enumerate the sets\label{ln:nozero-6}
    \begin{align*}
        \mathcal{L}_{\bm{r}} \coloneqq &\left\{ \vec{a} \in C_1^{n/2} \times C_2^{n/2} \; \middle| \; \vec{a} \cdot \vec{x} \equiv \bm{r} \pmod{\bp} \right\} \\
        \mathcal{R}_{\bm{r}} \coloneqq &\left\{ \vec{b} \in C_2^{n/2} \times C_1^{n/2} \; \middle| \; \vec{b} \cdot \vec{x} \equiv -\bm{r} \pmod{\bp} \right\}
    \end{align*}
    Halt if $|\mathcal{L}_{\bm{r}}|$ or $|\mathcal{R}_{\bm{r}}| \geq \frac{|\mathcal{L}|}{p_{max}}\cdot n^{\omega(1)}$.
    \tcp*{see~\Cref{lem:balancedSB-without0}, proof of runtime.} 
	Enumerate and sort the lists $L \coloneqq \mathcal{L}_{\bm{r}} \cdot
    \vec{x}$ and $R \coloneqq \mathcal{R}_{\bm{r}} \cdot \vec{x}$.\label{ln:nozero-7}\\
    Use Meet-in-the-Middle to search for a solution pair in $L \times R$.\label{ln:nozero-8}
    \caption{Outline of a representation technique-style algorithm for instances
    of Subset Balancing on coefficient sets of the form $[\pm d]$ with balanced solutions. Certain details are deferred to the proof of \Cref{lem:balancedSB-without0} for readability.}
    \label{alg:BalancedSBWithout0}
\end{algorithm}

\begin{lemma}[Solving Subset Balancing with \texorpdfstring{$C = [\pm d]$, $d > 2$}{} on instances with $\eps$-balanced solutions]
    \label{lem:balancedSB-without0}
    Let $C = [\pm d]$, $d > 2$ be a coefficient set. Then there exists an algorithm that solves any Subset Balancing instance $\vec{x}$ on $C$ that does \emph{not} admit any $\epsilon$-unbalanced solution, with high probability, in time
    \[
        \Os\left(|C|^{(1/2 - \delta)n}\right),
    \]
    for some positive constants $\epsilon, \delta$ that depend only on $|C|$.
\end{lemma}

\begin{proof}[Proof of correctness for \Cref{lem:balancedSB-without0}]
    Let $C$ and $\vec{x}$ be as in the lemma statement 
    and choose a pair of good coefficient set factors $C_1$ and $C_2$ for $C$.
    (Good coefficient set factors for $C$ are guaranteed to exist by
    \Cref{lem:good-coeff-sets}.) Let $\gamma \coloneqq \gamma(C_1, C_2)$ be the largest constant such that 
    \begin{equation}
        \label{eq:gamma-without0}
        \frac{|\mathcal{L}|}{|P(\vec{z}, C_1, C_2)|} = o(|C|^{n/2} \cdot 2^{-\gamma n})
    \end{equation}
    holds for a perfectly balanced solution profile $\vec{z} \in C^n$. Note that $\gamma > 0$ by the definition of good coefficient set factors (\Cref{defn:good-factors}).

    As usual, because our algorithm will not return false positives, we can assume the existence of a solution without loss of generality. However, we cannot assume the existence of a \emph{perfectly balanced} solution vector. Fortunately this is unnecessary. Letting $\vec{z}$ denote any perfectly balanced solution vector, we observe that
    \[
        \lim_{\epsilon \rightarrow 0} \; \; \min_{\substack{\epsilon\text{-balanced}\\\vec{c} \in C^n}} \; |P(\vec{c}, C_1, C_2)| = |P(\vec{z}, C_1, C_2)|.
    \]
    That is, as $\epsilon$ approaches $0$, the minimum number of solution pairs for an $\epsilon$-balanced coefficient vector approaches that of a perfectly balanced coefficient vector. The exact size of $|P(\vec{c}, C_1, C_2)|$ is a function of the entropy the solution profile of $\vec{c}$. Thus we can choose our constant $\epsilon > 0$ small enough to satisfy
    \begin{equation}
        \label{eq:eps-size-condn-without0}
        \frac{|\mathcal{L}|}{|P(\vec{c}, C_1, C_2)|} = o(|C|^{n/2} \cdot 2^{-\frac{3}{4} \gamma n})
    \end{equation}
    for all $\epsilon$-balanced coefficient vectors $\vec{c}$.

    We can then run Algorithm \ref{alg:BalancedSBWithout0} on all $\epsilon$-balanced solution profiles $\pi$.
    Thus without loss of generality we proceed under the assumption that we have correctly guessed an $\epsilon$-balanced solution profile $\pi$ corresponding to some solution vector $\vec{c}$.

    By \Cref{lem:perf-mixing2}, if
    \[
        |\{P_{\mathcal{L}} \cdot \vec{x}\}| \leq |P(\vec{c},C_1,C_2)| \cdot 2^{-\frac{1}{2} \gamma n},
    \]
    then $\vec{x}$ has at least $2^{\frac{1}{2} \gamma n}$ solution vectors. By applying the sampling approach of \Cref{lem:SB-many-solutions} with $\epsilon = \frac{1}{2} \gamma$, Algorithm \ref{alg:BalancedSBWithout0} recovers a solution with high probability in time 
    \[
        \Os\left(|C|^{n/2} \cdot 2^{-\frac{1}{4}\gamma)n}\right)
    \]
    in this case.
    It remains to show that the remaining steps of Algorithm \ref{alg:BalancedSBWithout0} recover a solution in the case that
    \begin{equation}
        \label{eq:solution-dist-without0}
        |\{P_{\mathcal{L}} \cdot \vec{x}\}| > |P(\vec{c},C_1,C_2)| \cdot 2^{-\frac{1}{2} \gamma n}
    \end{equation}
    
    When we run Algorithm \ref{alg:BalancedSBWithout0}, we set
    \begin{equation}
        \label{eq:pmax-without-0}
        p_{max} \coloneqq |P(\vec{c}, C_1, C_2)| \cdot 2^{-\frac{1}{2} \gamma n}.
    \end{equation}

    Next we apply \Cref{lem:prime-dist-main} using the sets 
    \begin{align*}
        &\{P_{\mathcal{L}} \cdot \vec{x}\}\text{ for }G\text{ and }\\
        &(\mathcal{L} \cdot \vec{x}) \cup (\mathcal{R} \cdot -\vec{x})\text{ for }Y.
    \end{align*}
    (Note that this choice for $Y$ controls the size of $\mathcal{L}_{\bm{r}}$ and $\mathcal{R}_{\bm{r}}$ simultaneously.) Let $b$ be a constant satisfying $\diam(\{P_{\mathcal{L}} \cdot \vec{x}\}) < 2^{n^b}$, and note that $b$ is indeed constant under our assumption that input integers have size $2^{n^{O(1)}}$. Because
    \[
        p_{max} = |P(\vec{c}, C_1, C_2)| \cdot 2^{-\frac{1}{2}\gamma n} \leq |\{P_{\mathcal{L}} \cdot \vec{x}\}|
    \]
    by \eqref{eq:solution-dist-without0}, the lemma condition that $p_{max} \leq O(|G|)$ is satisfied. As a result, we have that with constant probability the set of good residue classes $\bm{R}$ has cardinality $|\bm{R}| > p_{max} \cdot n^{-(b + 1)}$. By the definition of $\bm{R}$ (\Cref{lem:prime-dist-main}), for any good residue class $\bm{r}$, we have
    \begin{itemize}
        \item Upper bounds on the size of $\mathcal{L}_{\bm{r}}$ and $\mathcal{R}_{\bm{r}}$:
        \begin{align*}
            |\mathcal{L}_{\bm{r}}| &\leq \frac{2(|\mathcal{L}|+|\mathcal{R}|)}{p_{max}} \cdot n^{b + 1} \leq \frac{4|\mathcal{L}|}{|P(\vec{c}, C_1, C_2)|} \cdot 2^{\frac{1}{2} \gamma n} \cdot n^{b + 1} \leq \Os(|C|^{n/2} \cdot 2^{-\frac{\gamma}{4} n}) \text{ and } \\
            |\mathcal{R}_{\bm{r}}| &\leq \frac{2(|\mathcal{L}|+|\mathcal{R}|)}{p_{max}} \cdot n^{b + 1} \leq \frac{4|\mathcal{L}|}{|P(\vec{c}, C_1, C_2)|} \cdot 2^{\frac{1}{2} \gamma n} \cdot n^{b + 1} \leq \Os(|C|^{n/2} \cdot 2^{-\frac{\gamma}{4} n}),
        \end{align*}
        by \eqref{eq:eps-size-condn-without0} and \eqref{eq:pmax-without-0}.
        \item Guaranteed existence of solution pair(s): $\mathcal{L}_{\bm{r}}$ contains at least 
        \[
            \frac{|\{P_{\mathcal{L}} \cdot \vec{x}\}|}{4 p_{max}} > 0
        \]
        elements of $\{P_{\mathcal{L}} \cdot \vec{x}\}$. Note that this implies that $\mathcal{R}_{\bm{r}}$ contains the corresponding elements of $\{P_{\mathcal{R}} \cdot \vec{x}\}$ by construction.
    \end{itemize}
    In this case, the final two steps of Algorithm \ref{alg:BalancedSBWithout0} recover a solution deterministically. This completes the proof of correctness for \Cref{lem:balancedSB-without0}.
\end{proof}

\begin{proof}[Proof of runtime for \Cref{lem:balancedSB-without0}]
    The guarantee of correctness for \Cref{lem:balancedSB-without0} is achieved by running Algorithm \ref{alg:BalancedSBWithout0} $\poly(n)$ times on all solution profiles that are $\epsilon$-balanced (for the constant $\epsilon > 0$ chosen in the proof of correctness above). Because the additional polynomial factors are absorbed by the $\Os$ notation, it remains to bound the runtime of Algorithm \ref{alg:BalancedSBWithout0} on an $\epsilon$-balanced solution profile.

    Lines \ref{ln:nozero-1}, \ref{ln:nozero-2}, \ref{ln:nozero-4} and
    \ref{ln:nozero-5} can be implemented in polynomial time. Line \ref{ln:nozero-3} implements the procedure described in
    \Cref{lem:SB-many-solutions}, which runs in time 
    $\Os\left(|C|^{n/2} \cdot 2^{-\frac{1}{4}\gamma)n}\right)$
    for the constant $\gamma > 0$ determined by the coefficient set factors $C_1$ and $C_2$, as described in the proof of correctness above.

    Lines~\ref{ln:nozero-7} and \ref{ln:nozero-8} can be easily implemented in time $\Os(|\mathcal{L}_{\bm{r}}| + |\mathcal{R}_{\bm{r}}|) = \Os(|\mathcal{L}|/p_{max})$, which is upper-bounded by the runtime of Line~\ref{ln:nozero-6}. Line~\ref{ln:nozero-6} itself can be implemented in time 
    \[
        \Os\left(\max\left\{p_{max}, \frac{|\mathcal{L}|}{p_{max}} \right\}\right)
    \]
    using the dynamic programming approach described in the proof of runtime for \Cref{lem:balanced-SB-with0s} to enumerate items that fall in the residue classes $\bm{r}$ and $-\bm{r} \pmod{\bp}$. We can bound the first term of this expression by noting that
    \begin{align*}
        p_{max} &< |P(\vec{c}, C_1, C_2)| \\
        &< 2^n \\
        &< |C|^{n/2} \cdot 2^{-n/4},
    \end{align*}
    where the first inequality follows from \eqref{eq:pmax-without-0}, the second, without loss of generality, from \Cref{obs:p-small-without0}, and the final inequality from the fact that $|C| \geq 6$ by assumption. Finally, we observe that $\frac{|\mathcal{L}|}{p_{max}} < |C|^{n/2} \cdot 2^{-\frac{\gamma}{4} n}$ by \eqref{eq:eps-size-condn-without0} and \eqref{eq:pmax-without-0}.

    Thus the runtime of Algorithm \ref{alg:BalancedSBWithout0} is at most $\Os(|C|^{(\frac{1}{2} - \delta)n})$ on an $\epsilon$-balanced solution profile, where $\epsilon, \delta > 0$ are constants that depend only on $|C|$. This completes the proof of runtime for \Cref{lem:balancedSB-without0}.
\end{proof} 

\subsection{ Breaking Meet-in-the-Middle for Subset Balancing on \texorpdfstring{$C = [\pm d], d > 2$}{} }
\label{subsec:runtime-opt-pm3}

\begin{theorem}
    \label{thm:main-without0}
    For any $d > 2$, there exists a constant $\delta > 0$ such that Subset Balancing on $C = [\pm d]$, can be solved with high probability in time
    \[
        \Os(|C|^{(1/2 - \delta)n})
    \]
\end{theorem}
\begin{proof}
    Choose $\epsilon > 0$ to satisfy the requirements for \Cref{lem:balancedSB-without0}. Then apply \Cref{lem:unbalanced,lem:balancedSB-without0}.
\end{proof}

Once again, the optimal value of the constant $\delta$ results from a messy trade-off between the two approaches. As a proof of concept, we maximize $\delta$ for $C = [\pm 3]$.

\begin{proposition}
    \label{prop:SB-pm3-without0}
    Subset Balancing on $C = [\pm 3]$ can be solved with high probability in time 
    \[
        \Os(|C|^{0.495n}).
    \]
\end{proposition}
\begin{proof}
    Similar to the approach taken in the proof of \Cref{prop:SB-pm2-with0}, we optimize runtime by running the faster of Algorithm \ref{alg:UnbalancedSB}, our algorithm for unbalanced Subset Balancing instances, and Algorithm \ref{alg:BalancedSBWithout0}, our algorithm for balanced Subset Balancing instances on $C = [\pm d], d > 2$, to search for solution vectors matching each possible solution profile $\pi$. The runtimes of both algorithms can be calculated explicitly as a function of the fixed solution profile $\pi$, so we can find the worst-case runtime by  taking the maximum over the $\poly(n)$ possible solution profiles.

	\subparagraph*{Case 1: Unbalanced Solutions. } If there exists a solution vector matching $\pi$, Algorithm~\ref{alg:UnbalancedSB} solves the instance in time
        \begin{equation}
            \label{eq:pm3-without0-runtime1}
            \Os\left(2^{H\left(\frac{\pi(-3)}{n}, \frac{\pi(-2)}{n}, \frac{\pi(-1)}{n}, \frac{\pi(1)}{n}, \frac{\pi(2)}{n}, \frac{\pi(3)}{n}\right) \frac{n}{2}} \right).
        \end{equation}
        \subparagraph*{Case 2: Balanced Solutions. } 
        For $C = [\pm 3]$, our choice for a pair of good coefficient set factors
    is $C_1 = \{0, 1\}$, $C_2 = \{-3, -2, 1, 2\}$.
    Now suppose that for our values of $\pi, C_1, C_2$, the inequality
        \begin{equation}
            \label{eq:pm3-without0-solnpairs-hypothesis}
            \frac{|\mathcal{L}(C_1, C_2)|}{|P(\vec{z}, C_1, C_2)|} < |C|^{n/2} \cdot 2^{-\alpha n}
        \end{equation}
        holds for some constant $\alpha > 0$ and an arbitrary coefficient vector $\vec{z}$ matching $\pi$. 
        In this case, we can solve the instance by applying \Cref{lem:SB-many-solutions} with $\epsilon = \frac{2\alpha}{3}$ to solve the instance in time
        \[
                \Os\left(|C|^{n/2} \cdot 2^{-\frac{\alpha}{3}n} \right)
            \text{ if }
            |\{P_{\mathcal{L}} \cdot \vec{x}\}| \leq |P(\vec{z}, C_1, C_2)| \cdot 2^{-\frac{2}{3}\alpha n},
        \]

        and we can likewise solve the instance in time
            $\Os\left(|C|^{n/2} \cdot 2^{-\frac{\alpha}{3}n} \right)$
        by executing Lines~\ref{ln:nozero-5}-\ref{ln:nozero-8} of Algorithm~\ref{alg:BalancedSBWithout0} if $|\{P_{\mathcal{L}} \cdot \vec{x}\}| > |P(\vec{z}, C_1, C_2)| \cdot 2^{-\frac{2}{3}\alpha}n$.

        We can solve \eqref{eq:pm3-without0-solnpairs-hypothesis} for $\alpha$ by plugging in the appropriate values: 
        \begin{align*}
            |\mathcal{L}(C_1, C_2)| &= |C_1^{n/2} \times C_2^{n/2}| = 2^{n/2} \cdot 4^{n/2} = 2^{\frac{3}{2}n},\\
            |P(\vec{z},C_1, C_2)| &= 2^{\pi(-2) + \pi(2)},\text{ and }\\
            |C|^{n/2} &= 6^{n/2},
        \end{align*}
        which yields
        \[
            \alpha < \frac{\pi(2) + \pi(-2)}{n} + \frac{\log_2(6)}{2} - \frac{3}{2},
        \]
        for a final runtime of 

        \begin{align}
            \Os\left(|C|^{n/2} \cdot 2^{-\frac{\alpha}{3}n} \right)
            &= \Os\left(6^{n/2} \cdot 2^{-\frac{1}{3}(\frac{\pi(2) + \pi(-2)}{n} + \frac{\log_2(6)}{2} - \frac{3}{2})n}\right)\\
            &= \Os\left(6^{n/2} \cdot 2^{-\frac{\pi(-2) + \pi(2)}{3}} \cdot 6^{- \frac{n}{6}} \cdot 2^{n/2}\right)\\
            &= \Os\left(6^{n/3} \cdot 2^{n/2} \cdot 2^{-\frac{\pi(-2) + \pi(2)}{3}}\right)
            \label{eq:pm3-without0-runtime2}
        \end{align}
        in this case.
    We observe that, for any fixed value of $(\pi(-2) + \pi(2))$, the Case 1
    runtime \eqref{eq:pm3-without0-runtime1} is maximized when
    \begin{align*}
        \pi(-2) = \pi(2)\text{ and }
        \pi(-3) = \pi(-1) = \pi(1) = \pi(3) = \frac{n - 2\pi(2)}{4},
    \end{align*}
    while the Case 2 runtime \eqref{eq:pm3-without0-runtime2} is unchanged, so
    we can make these assumptions on $\pi$ without loss of generality. By taking
    the base-2 logarithm of \eqref{eq:pm3-without0-runtime1} and
    \eqref{eq:pm3-without0-runtime2} and writing both equations in terms of
    $\beta \coloneqq \frac{\pi(2)}{n}$, we can balance our runtime for worst-case $\pi$ by solving the following optimization problem:

    \[
        \max_{\beta \in [0,0.5]} \min \left[
        \frac{1}{2} \cdot H\left(\frac{1 - 2\beta}{4}, \beta, \frac{1 - 2\beta}{4}, \frac{1 - 2\beta}{4}, \beta, \frac{1 - 2\beta}{4}\right), \;
            \frac{\log_2(6)}{3} + \frac{1}{2} - \frac{2\beta}{3}
        \right].
    \]
    Computer evaluation reveals that the expression is maximized at $\beta \approx 0.1232$ and the final runtime is upper-bounded by
    \[
            \Os(2^{1.27955n}) = \Os\left(|C|^{0.495n}\right).
        \qedhere
    \]
\end{proof}

\section{Improved Algorithm for Equal Subset Sum}\label{sec:ess}

The current fastest algorithm for Equal Subset Sum (Subset Balancing on $C = [-1 : 1]$) is due to Mucha, Nederlof, Pawlewicz, and Węgrzycki, who proved the following result:

\begin{lemma}[cf. \cite{mucha2019equal}, Theorems 3.3 and 3.4]\label{lem:mucha}
    Given an Equal Subset Sum instance that admits a solution with solution profile $\pi$, the solution can be recovered in time
    \begin{displaymath}
        \Os\left( \min 
            \left\{
                \binom{n/2}{\frac{\pi(0)}{2}} \binom{\frac{\pi(1) + \pi(-1)}{2}}{\frac{\pi(1)}{2}},
                \binom{\pi(1) + \pi(-1)}{\pi(1)}
            \right\}
        \right)
    \end{displaymath}
    with high probability.
\end{lemma}
This expression is maximized when $\pi(-1) + \pi(1) \approx 0.773n$, in which case it evaluates to $\Oh(1.7088^n)$. Because the algorithm returns no false positives and can be run on all $\poly(n)$ solution profiles in parallel, this is also the worst-case runtime for Equal Subset Sum. In this section we apply several additional techniques, including a mixing dichotomy and the method of compatibility certificates, to improve this result.

\subsection{Preprocessing}
\label{subsec:ess-preprocessing}

Let $\vec{x} = (x_1,\ldots,x_n)$ be an instance of Equal Subset Sum. As in the previous section, we assume without loss of generality that $\vec{x}$ admits a solution vector $\vec{c} \in C^n$ and that $\pi$ is the solution profile corresponding to $\vec{c}$. We further assume that $\vec{c}$ is a solution with \emph{minimal support}; that is, $\pi$ maximizes $\pi(0)$ over all (nonzero) solution vectors for $\vec{x}$. By~\cref{lem:rerandomization}, we can assume that
\begin{equation*}
    \pi(1) = \pi(-1)
\end{equation*}
without loss of generality at the cost of an extra $n^{O(1)}$ factor in the runtime. Finally, we assume that 
\begin{equation*}
    \pi(0) \le \frac{n}{3}
\end{equation*}
as otherwise the problem can be solved in time $\Oh^*(1.7^n)$ by~\cref{lem:mucha}. 

\subsection{Good Solution Pairs} 
\label{subsec:ess-goodpairs}

The most straightforward set of solution pairs for Equal Subset Sum given a
fixed solution vector $\vec{c}$ is the set $\{(\vec{a}, \vec{b}) \in \{0, 1\}^n
\times  \{0, 1\}^n \; | \; \vec{a} - \vec{b} = \vec{c}\}$. Our set of solution
pairs is based on this but slightly more complicated, so we introduce a new definition.

\begin{definition}[Good Solution Pair]\label{def:good-pair}
    With respect to a solution vector $\vec{c}$ and a non-negative constant $\epsilon < 1/12$, the set of \emph{good solution pairs} $G(\vec{c}) \subset [0:2]^n \times [0:2]^n$ consists of all pairs $(\vec{a}, \vec{b})$ satisfying
    \begin{enumerate}[label=(\roman*)]
        \item\label{def:good-pair-1} $|\vec{a}^{\,-1}(1)| = |\vec{b}^{\,-1}(1)| = n/2$ and $|\vec{a}^{\,-1}(2)| = |\vec{b}^{\,-1}(2)| = \eps n$;
        \item\label{def:good-pair-2} $\vec{a}-\vec{b} = \vec{c}$; and
        \item\label{def:good-pair-3} for every $i \in [n]$, if $\vec{c}_i = 0$ then $(\vec{a_i},\vec{b_i}) \in \{(0,0),(1,1)\}$.
    \end{enumerate}
\end{definition}

Here Condition~\ref{def:good-pair-1} specifies the number of $0$'s, $1$'s, and
$2$'s in the vectors $\vec{a}$ and $\vec{b}$. $\eps$ in particular controls the
number of $2$'s: increasing $\epsilon$ trades off an improvement in runtime due
to increased $|G(\vec{c})|$ and a slowdown due to more complicated solution
recovery. Condition~\ref{def:good-pair-2} ensures that $(\vec{a}, \vec{b})$ corresponds to a solution vector. Condition~\ref{def:good-pair-3} is an additional technical requirement ruling out indices at which $(\vec{a}_i, \vec{b}_i) = (2, 2)$. It will allow us to prove a perfect mixing dichotomy: below, we will show that if two different solution pairs ``collide'' with respect to $\vec{x}$ the ESS instance admits a solution with small support that can be recovered relatively quickly. 

\begin{observation}
    For a fixed solution $\vec{c}$ with $\pi(1) = \pi(-1)$ and constant $\epsilon < 1/12$, the number of good solution pairs is
    \begin{align}\label{eq:psi}
        p_{max} \coloneqq |G(\vec{c})| 
        = \Theta^*\left(\binom{\pi(1)}{\eps n}^2 \cdot 2^{\pi(0)} \right).
    \end{align}
\end{observation}
\begin{proof}
    For every index $i \in [n]$, the value $\vec{c}_i$ determines $\vec{a}_i - \vec{b}_i$ by definition. We observe that
    \begin{align*}
        \text{If } \vec{c}_i &= 1 \text{, then } (\vec{a}_i, \vec{b}_i) \in \{(1, 0), (2, 1)\}\text{.}\\
        \text{If } \vec{c}_i &= -1\text{, then }(\vec{a}_i, \vec{b}_i) \in \{(0, 1), (1, 2)\}\text{.}\\
	\text{If } \vec{c}_i &= 0\text{,  then }(\vec{a}_i, \vec{b}_i) \in \{(0, 0), (1, 1)\}\text{ by Condition~\ref{def:good-pair-3}.}
    \end{align*}

    Note that the $\epsilon n$ appearances of $2$ in $\vec{a}$ must occur at indices $i$ with $c_i = 1$. Thus the factor $\binom{\pi(1)}{\epsilon n}$ represents the number of ways to choose $\epsilon n$ positions for $2$'s in $\vec{a}$ from the $\pi(1)$ indices where $\vec{c}_i = 1$. Similarly, a second factor $\binom{\pi(1)}{\epsilon n} = \binom{\pi(-1)}{\epsilon n}$ represents the number of ways to choose $\epsilon n$ positions for $2$'s in $\vec{b}$ among the $\pi(-1)$ indices where $\vec{c}_i = -1$.

	Once the positions for $2$'s in $\vec{a}$ and $\vec{b}$ are fixed, the value of $(\vec{a}_i, \vec{b}_i)$ is determined at every index $i$ for which $\vec{c}_i = 1$ or $-1$. Thus it remains to count the number of ways to assign $(\vec{a}_i, \vec{b}_i)$ at indices $i$ for which $\vec{c}_i = 0$. By Condition~\ref{def:good-pair-1}, all but $\pi(0)/2$ of the $1$'s in the vector $\vec{a}$ and all but $\pi(0)/2$ of the $1$'s in the vector $\vec{b}$ are accounted for, so half of the indices at which $c_i = 0$ must have $(\vec{a}_i, \vec{b}_i) = (0, 0)$ and half must have $(\vec{a}_i, \vec{b}_i) = (1, 1)$. Thus there are $\Theta^*(2^{\pi(0)})$ valid choices by Stirling's approximation \eqref{eq:stirling}. Multiplying this term with the factors for the choices of $2$'s in $\vec{a}$ and $\vec{b}$ completes the proof of the observation.
\end{proof}

Our definition of $G(\vec{c})$ allows us to prove a new perfect mixing dichotomy: either $G(\vec{c})$ mixes perfectly with respect to $\vec{x}$, or our assumption that the solution with minimal support has $\pi(0) \leq n/3$ is violated, in which case we can solve the instance using \Cref{lem:mucha}.

\begin{lemma}[Perfect Mixing Dichotomy for Equal Subset Sum]
    \label{lem:ess-perfect-mixing}
    For a minimal-support solution $\vec{c}$ with $\pi(0) \le n/3$ and $\pi(1) =
    \pi(-1)$ and non-negative constant $\eps < 1/12$, for every two distinct
    good pairs $(\vec{a}, \vec{b}), (\vec{a}\,', \vec{b}\,') \in G(\vec{c})$, we have
    \[
        \vec{a} \cdot \vec{x} \neq \vec{a}\,' \cdot \vec{x}.
    \]
\end{lemma}
\begin{proof}
    Assume for the sake of contradiction that $\vec{a} \cdot \vec{x} =
    \vec{a}\,' \cdot \vec{x}$ for $\vec{a}$ and $\vec{a}\,'$ as in the lemma statement. As a consequence of \Cref{def:good-pair} we have that
    \[
        \vec{c}\,' \coloneqq \vec{a}-\vec{a}\,' \in [-1 : 1]^n
    \]
is a solution for Equal Subset Sum. (Note that this crucially relies on Condition~\ref{def:good-pair-3}, as without it an index $i$ at which $c_i = 0$ and $a_i - a'_i = 2$ or $-2$ would be possible.) 
    
Because $(\vec{a},\vec{b})$ and $(\vec{a}\,',\vec{b}\,')$ are good solution pairs for
    $\vec{c}$ by assumption, for every index $i \in [n]$ with $c_i = 1$ we
    have that $a_i = a_i' = 1$ unless $a_i = 2$ or $a'_i
    = 2$, which is true for at most $2 \eps n$ indices. Likewise $a_i =
    a_i' = 0$ at all but $2\eps n$ indices $i \in [n]$ where $c_i =
    -1$. Thus the number of $i \in [n]$ for which $a_i - a_i' = 0$ is at
    least
    \[
        \pi(1) + \pi(-1) - 4\eps n > 2n/3 - n/3 = n/3
    \]
    by our assumptions on $\pi(0)$ and $\epsilon$. However, this implies that $\vec{a} - \vec{a}'$ is a solution with support less than $\vec{c}$, contradicting our assumption that $\vec{c}$ is a solution of minimal support.
\end{proof}

\subsection{ Algorithm for Equal Subset Sum }

As usual, the algorithm below is parameterized by the solution profile $\pi$ of the solutions it attempts to recover. To get our final result for ESS we will use it only when it is faster than \Cref{lem:mucha}.

The algorithm follows the general pattern of the representation technique.
However, because we have chosen to use partial solution vectors supported on
$[0:2]^n$, the list $\calS \times \calS$ may contain pseudosolutions (pairs
$(\vec{a}, \vec{b}) \in \calS \times \calS$ such that $(\vec{a} - \vec{b}) \cdot \vec{x} = 0$, but $\vec{a} - \vec{b} \notin [-1:1]^n$). To search for true solutions, for every possible value $i$ we generate a sublist $L_i$ containing all partial solution vectors satisfying $\vec{v} \cdot \vec{x} = i$. Then we run a subroutine that searches for true solutions given two lists of partial solution vectors, a problem that is reminiscent of solving sparse instances of Orthogonal Vectors. Our main result for this problem is given in \Cref{thm:compatibility}. The proof uses the method of compatibility certificates and is presented in \Cref{sec:compatibility-test}.

\begin{algorithm}[ht!]
    \DontPrintSemicolon
    \nlnonumber
    \textbf{function} $\textsf{EqualSubsetSum}(\vec{x}, \pi, \epsilon)$:\\
    Sample a prime $\bp \sim [p_{max} : 2 p_{max}]$ and residue class $\bm{r}
    \sim [\bp]$ for $p_{max}$ as defined in \eqref{eq:psi}.\label{ln:ess-1}\\
    Enumerate the set\label{ln:ess-2}
        \begin{align*}
            \calS &\coloneqq \{ \vec{v} \in [0:2]^n \; | \; |\vec{v}^{\,-1}(1)| = n/2,
            |\vec{v}^{\,-1}(2)| = \epsilon n, \vec{v} \cdot \vec{x} \equiv \bm{r} \pmod{\bp} \}.
        \end{align*}
        If $|\calS| \geq \Ex[|\calS|] \cdot n^{\omega(1)}$
        , halt and return ``\emph{failure}''.
        
    Enumerate the list $L_i \coloneqq \{ \vec{v} \in \calS \; | \; \vec{v} \cdot \vec{x} = i \}$ for every value $i$ that yields a nonempty list. \label{ln:ess-3}\\
    For each $L_i$, use \Cref{thm:compatibility} to decide if there exists a good solution pair $(\vec{a}, \vec{b}) \in L_i \times L_i$.\label{ln:ess-4}
    \caption{Outline of our algorithm for Equal Subset Sum. Implementation
    details omitted from this figure are specified in the proofs of correctness
and runtime (\Cref{lem:ess-runtime-part1,lem:ess-correctness}).}
    \label{alg:ESS}
\end{algorithm}

We start by estimating $\Ex[|\calS|]$.

\begin{observation}\label{lem:exp-list-length}
    Let $\calS$ be defined with respect to a solution profile $\pi$ with $\pi(0) \leq n/3$ and $\pi(1) = \pi(-1)$, a non-negative constant $\epsilon < 1/12$, a random prime $\bp \sim [p_{max} : 2p_{max}]$ and a random residue class $\bm{r} \sim [\bp]$ as in Algorithm \ref{alg:ESS}. Then
    \begin{equation}
         \label{eq:exp-list-length-ESS}
         \Ex[|\calS|] = \Os \left( \frac{2^n \cdot \binom{n/2}{\epsilon n}} {p_{max}} \right)
    \end{equation}
\end{observation}
\begin{proof}
    The size of the set 
    \begin{equation*}
        \label{eq:ess-C0}
        \calS_0 \coloneqq \{ \vec{v} \in [0:2]^n \mid |\vec{v}^{\,-1}(1)| = n/2,
        |\vec{v}^{\,-1}(2)| = \eps n\}
    \end{equation*}
    is $O^*(2^n \cdot \binom{n/2}{\epsilon n})$ by Stirling's approximation.
    Because $\bp \geq p_{max}$ by definition and every element of $\calS_0$
    falls into $\calS$ with probability $1 / \bp$, the result follows by linearity of expectation.
\end{proof}

\begin{lemma}[Runtime of Algorithm \ref{alg:ESS}]
    \label{lem:ess-runtime-part1}
    Given an input instance $\vec{x}$, a solution profile $\pi$ with $\pi(0) \leq n/3$ and $\pi(1) = \pi(-1)$, and a non-negative constant $\epsilon < 1/12$, Algorithm \ref{alg:ESS} runs in time 
    \begin{equation}
        \label{eq:ess-runtime-variables}
        \Os\left( \frac{2^n \cdot \binom{n/2}{\eps n} \cdot 2^{c(\eps) n}}{p_{max}} \right),
    \end{equation}
    where $p_{max}$ is defined as in \eqref{eq:psi} and $c(\eps)$ is defined as in \eqref{eq:c-certificate}.
\end{lemma}
\begin{proof}
 Line~\ref{ln:ess-1} needs polynomial time. Note that we can compute $p_{max}$ for a given $\pi$ even if $\vec{c}$ is unknown or no solution vector corresponding to $\pi$ exists.

    For Lines~\ref{ln:ess-2} and \ref{ln:ess-3}, we can enumerate $\calS$ in
    time $\Os(|\calS|)$ using the dynamic programming approach described in the proof of runtime for \Cref{lem:balanced-SB-with0s}. The algorithm halts if 
    \[
        |\calS| \geq \Ex[|\calS|] \cdot n^{\omega(1)} = \frac{2^n \cdot \binom{n/2}{\eps n}}{p_{max}} \cdot n^{\omega(1)},
    \]
    where we use \eqref{eq:exp-list-length-ESS} to bound $\Ex[|\calS|]$, so this
    is an upper bound on the runtime of this step. Once we have enumerated
    $\calS$, every sublist $L_i$ can also be enumerated in time $\Os(|\calS|)$
    by sorting the vectors in $\calS$ by dot product with $\vec{x}$. 

    In Line~\ref{ln:ess-4}, we apply \Cref{thm:compatibility} to every pair
    $(L_i, L_i)$, which takes time $O(2^{c(\eps) n} \cdot |\calS|)$ in total for
    all list pairs. Plugging in the upper bound on $\calS$ enforced by the algorithm yields \eqref{eq:ess-runtime-variables}.
\end{proof}

\begin{lemma}[Correctness of Algorithm \ref{alg:ESS}]
    \label{lem:ess-correctness}
    Given an input instance $\vec{x}$, a solution profile $\pi$ with $\pi(0) \leq n/3$ and $\pi(1) = \pi(-1)$, and a non-negative constant $\epsilon < 1/12$, if there exists any solution vector $\vec{c}$ corresponding to $\pi$, Algorithm \ref{alg:ESS} recovers it with probability $n^{-O(1)}$.
\end{lemma}
\begin{proof}
    Fix $\vec{x}$, $\pi$, $\epsilon$ and $\vec{c}$ as in the lemma statement.
    For any good solution pair $(\vec{a}, \vec{b})$ we have $(\vec{a} - \vec{b})
    \cdot \vec{x} = \vec{c} \cdot \vec{x} = 0$ by definition, so $\vec{a} \cdot
    \vec{x} = \vec{b} \cdot \vec{x}$ and thus $\vec{a} \in \calS$ if and only if
    $\vec{b} \in \calS$. If $\vec{a}, \vec{b} \in \calS$ and $|\calS| <
    \Ex[|\calS|] \cdot n^{\omega(1)}$, then the algorithm enumerates $\calS$ and $\vec{a}, \vec{b} \in L_i$ for $i = \vec{a} \cdot \vec{x}$. Finally, in this case \Cref{thm:compatibility} recovers $(\vec{a}, \vec{b})$ with probability $n^{-O(1)}$.

    It remains to show that, with probability $n^{-O(1)}$, $\vec{a} \in \calS$
    for some solution pair $(\vec{a}, \vec{b})$ and the list $\calS$ satisfies
    $|\calS| < \Ex[|\calS|] \cdot n^{\omega(1)}$. This follows from \Cref{lem:prime-dist-main} with
    \begin{align*}
        G &= \{ \vec{a} \cdot \vec{x} \; | \; (\vec{a}, \vec{b}) \in G(\vec{c}) \},\\
        Y &= \calS_0 \cdot \vec{x} = \{ \vec{v} \in [0:2]^n \; \; | \; |\vec{v}^{\,-1}(1)| =
        n/2, |\vec{v}^{\,-1}(2)| = \epsilon n\} \cdot \vec{x},
   \end{align*}
    using $p_{max}$ as the integer bound represented by the same symbol in the
    lemma statement. Importantly, we have $|G| = |G(\vec{c})|$ by our perfect
    mixing dichotomy (\Cref{lem:ess-perfect-mixing}), $p_{max} \leq O(|G|)$ by
    definition \eqref{eq:psi}, and $Y \supseteq G$ with $|Y| = \Os(2^n \cdot \binom{n/2}{\eps
    n})$. Thus the set of good residue classes $\bm{R}$, which contains choices
    for $\bm{r}$ such that $\calS$ contains a solution pair and 
    \[
        |\calS| \leq \frac{|Y|}{p_{max}} n^{O(1)} = \Ex[|\calS|] \cdot n^{O(1)},
    \]
    is an $n^{-O(1)}$-fraction of all choices for $\bm{r}$ with constant probability over the choice of $\bp$ and $\bm{r}$.

    The lemma follows from multiplying the constant probability that \Cref{lem:prime-dist-main} obtains with the probability that we choose a good residue class in that event and the chance that \Cref{thm:compatibility} recovers a solution. 
\end{proof}

\subsection{ Optimizing Runtime }

\begin{theorem}\label{thm:ess}
    Equal Subset Sum can be solved in time $\Os(1.7067^n)$ with high probability.
\end{theorem}

\begin{proof}
    Recall from \Cref{subsec:ess-preprocessing} that we can assume without loss of generality that $\pi(1) = \pi(1)$ and $\pi(0) \leq n/3$. Using these assumptions to rewrite the first term of \Cref{lem:mucha}, we have that we can solve Equal Subset Sum in time
    \begin{equation}\label{eq:tradeoff-1}
        \Os\left( 
                \binom{n/2}{\frac{\pi(0)}{2}} \binom{\frac{\pi(1) + \pi(-1)}{2}}{\frac{\pi(1)}{2}}
        \right)
        = \Os\left(2^{\frac{1}{2}H(\frac{\pi(0)}{n}) + \frac{\pi(1)+\pi(-1)}{2}}\right)
    \end{equation}
    using existing algorithms, where the second term simplifies the binomials using Stirling's Approximation \eqref{eq:stirling}. 

    By~\cref{lem:ess-runtime-part1} we know that Equal Subset Sum can be solved in
    \begin{equation}\label{eq:tradeoff-2}
        \Os\left( \frac{2^n \cdot \binom{n/2}{\eps n} \cdot 2^{c(\eps) n}}{p_{max}} \right)
        = \Os\left(
            2^{n + H(2\eps)\frac{n}{2} + c(\eps)n - 2 H\left(\frac{\eps n}{\pi(1)}\right)\pi(1) - \pi(0)}
        \right)
    \end{equation}
    where $p_{max}$ and $c$ are defined as in \eqref{eq:psi} and \eqref{eq:c-certificate}, respectively.\footnote{Note that at $\eps = 0$ this expression is equivalent to $\Os(\binom{\pi(1) + \pi(-1)}{\pi(1)})$, the other term in \Cref{lem:mucha}. To directly verify that~\eqref{eq:tradeoff-2} improves upon~\cref{lem:mucha}, one can check that for sufficiently small $\eps > 0$, \eqref{eq:tradeoff-2} is always smaller than this term near the worst-case, when $\pi(1) + \pi(-1) \approx 0.773n$.}

    To optimize the exponent in the runtime we can take the base-2 logarithm of \eqref{eq:tradeoff-1} and \eqref{eq:tradeoff-2}, write both expressions in terms of the single parameter $p = \frac{\pi(0)}{n}$, and choose $\epsilon < 1/12$ to minimize the constant in the exponent over worst-case $\pi$:
    \[
        \max_{p \in [0,1)} \min_{\eps \in [0,1/12)} 
        \Os\left(\min\left\{
            1 + \frac{H(2\eps)}{2} + c(\eps) 
            - 2 H\left(\frac{2 \eps}{1-p}\right)\frac{1-p}{2} - p,\;\;
            \frac{H(p) + 1-p}{2}
        \right\}\right)
    \]
    Computer evaluation reveals that this expression is maximized at $p = 0.22266$ at which point the best choice of $\epsilon$ is $\eps = 0.04493$ and the running time is $\Os(2^{0.771167n}) = \Os(1.7067^n)$ as desired.
\end{proof}

\section{Compatibility Testing}
\label{sec:compatibility-test}


In the previous section, we presented an algorithm for Equal Subset Sum that enumerated lists of partial solutions and searched the lists for a solution pair. However, some partial solution pairs were pseudosolutions: pairs of vectors $(\vec{a}, \vec{b})$ satisfying $(\vec{a} - \vec{b}) \cdot \vec{x} = 0$, but $\vec{a} - \vec{b} \not \in [-1 : 1]^n$. To solve this problem we referred to  a subroutine that could search two lists of vectors for a ``compatible pair'' satisfying $\vec{a} - \vec{b} \in [-1 : 1]^n$. This section presents that subroutine. 

An alternative approach would be to reduce our problem to an instance of Orthogonal Vectors. However, existing worst-case algorithms for Orthogonal Vectors run in time quadratic in the size of the input lists, which is too slow for our purposes. Instead, we must take advantage of the fact that our input vectors are effectively very sparse, since we need only worry about the $\epsilon n$ indices at which $\vec{a}_i$ or $\vec{b}_i = 2$. Exploiting this allows us to get an algorithm that approaches linear time in the length of the input lists as $\epsilon$ goes to $0$. Our algorithm builds on ideas from \cite{NederlofW21}, in which a similar approach was used. 

\begin{theorem}
    \label{thm:compatibility}
    Fix a constant $0 \le \eps \le 1/4$ and let $\inputA,\inputB \subseteq [0:2]^d$ be two sets of vectors such that for every $\vec{a} \in \inputA$ and $\vec{b} \in \inputB$ it holds that
    \begin{align*}
        |\vec{a}^{\,-1}(2)| = |\vec{b}^{\,-1}(2)| = \eps d \text{ and } |\vec{a}^{\,-1}(1)| =
        |\vec{b}^{\,-1}(1)| = d/2.
    \end{align*}
    There exists an algorithm that recovers $(\vec{a}, \vec{b}) \in \inputA \times \inputB$ such that $\vec{a}-\vec{b} \in [-1 : 1]^d$ with high probability,
    if such a pair exists, and runs in time $\Oh\left( 2^{c(\eps) d} \cdot \left(|\inputA| + |\inputB| \right)\right)$, where
    \begin{equation}
        \label{eq:c-certificate}
        c(\eps) \coloneqq
        (1-\eps) \cdot 
        H\left( \frac{\eps}{1-\eps} \right) +
        \left(1/2+\eps\right) \cdot H\left(
            \frac{4\eps}{1+2\eps}
        \right) - H(2\eps) +
     o(1).
    \end{equation}
\end{theorem}

In this section, the inputs to our problem are the sets $A$ and $B$. When we apply \Cref{thm:compatibility} in our algorithm for Equal Subset Sum, $|A|$ and $|B|$ will be \emph{exponential} in $n$, the number of inputs to the Equal Subset Sum instance, and the vectors in $A$ and $B$ will have dimension $n$. However, in this section we write $d$ for the dimension of the vectors in $A$ and $B$ to avoid confusion with the size of the input ($|A| + |B|$).

\subsection{Preprocessing}\label{sec:preprocessing}

Consider two sets of vectors $\inputA$ and $\inputB$ as in the theorem statement. Our algorithm will not return false positives so we assume without loss of generality that there exists some solution $(\vec{a}, \vec{b}) \in \inputA \times \inputB$ with $\vec{a} - \vec{b} \in [-1 : 1]^d$. 

We can divide the task of finding a compatible pair of vectors $(\vec{a}, \vec{b})$ into subproblems by partitioning the set of indices $[d]$ into multiple parts of equal cardinality and checking the compatibility of each set of indices separately. Let $\ell$ be a large constant that we will set later. Randomly partition $[d]$ into $\ell$ sets of equal size,\footnote{We can assume that $d$ is divisible by $\ell$ without loss of generality using simple padding arguments.} so
\[
     [d] = U_1 \sqcup \ldots \sqcup U_\ell.
\]
We first observe that with sufficiently high probability, the $0$, $1$, and $2$ indices of $\vec{a}$ and $\vec{b}$ are partitioned evenly between the $U_i$:

\begin{observation}\label{obs:support}
    With probability $d^{-\Oh(\ell)}$, for every $i \in [\ell]$ it holds that
	\begin{equation}\label{eq:support}
        |U_i \cap \vec{a}^{\,-1}(2)| = |U_i \cap \vec{b}^{\,-1}(2)| = \frac{\epsilon d}{\ell} \text{ and } |U_i \cap \vec{a}^{\,-1}(1)| = |U_i \cap \vec{b}^{\,-1}(1)| = \frac{d}{ 2\ell}.
    \end{equation}
\end{observation}
\begin{proof}
    It suffices to show that for any set $X \subseteq [d]$, the event
    \begin{equation}
        |U_i \cap X| = \frac{|X|}{\ell}
        \label{eq:Ui-good-distribution}
    \end{equation}
    holds for every $i \in [\ell]$ with probability $d^{-\Oh(\ell)}$. 

    Once this statement is proved, we can complete the observation by repeatedly applying the statement above as follows: first, apply \eqref{eq:Ui-good-distribution} with $X = \vec{a}^{\,-1}(1)$. Conditioning on the event that the indices $\vec{a}^{\,-1}(1)$ are distributed evenly over a random $\ell$-partition of $U$, we can consider the event that the set of indices $\vec{a}^{\,-1}(2)$ are distributed evenly over a random $\ell$-partition of $(U \setminus \vec{a}^{\,-1}(1))$. Conditioning on the previous two events, we can then consider the distribution of $\vec{b}^{\,-1}(1) \cap \vec{a}^{\,-1}(1)$, $\vec{b}^{\,-1}(1) \cap \vec{a}^{\,-1}(2)$, and $\vec{b}^{\,-1}(1) \cap ([d] \setminus (\vec{a}^{\,-1}(1) \cup \vec{a}^{\,-1}(2))$ over the relevant subsets of $U$. Finally, we can do likewise for $\vec{b}^{\,-1}(2)$. Because each one of these events occurs with probability $d^{-\Oh(\ell)}$ conditioned on the previous events, the entire sequence of events occurs with probability $d^{-\Oh(\ell)}$.

    Returning to the proof of \eqref{eq:Ui-good-distribution}: the number of partitions of $U$ such that \eqref{eq:Ui-good-distribution} holds for all $i \in [\ell]$ is
    \begin{displaymath} 
    \binom{|X|}{\frac{|X|}{\ell},\ldots,\frac{|X|}{\ell}} \cdot \binom{|U \setminus X|}{\frac{|U \setminus X|}{\ell},\ldots,\frac{|U \setminus X|}{\ell}}
    \end{displaymath}
    which is
    \begin{displaymath}
        \Theta^*(2^{H(1/\ell,\ldots,1/\ell) d}) = \Theta^*(\ell^d).
    \end{displaymath}
    by Stirling's approximation for multinomial coefficients \eqref{eq:stirling-multinomial}. On the other hand, the total number of partitions of $U$ into $\ell$ sets is $\ell^{d}$, so \eqref{eq:Ui-good-distribution} holds with probability at least $d^{-O(1)}$. 
\end{proof}

For every $i \in [\ell]$, we will randomly select a family of partial
certificates $\pCert_i$ intended to certify if two vectors are compatible on the
indices in $U_i$. Each element of $\pCert_i$ is a tuple containing two subsets
of $U_i$. Formally, for $i \in [\ell]$, we select
\begin{equation}
    \label{eq:pCertSelection}
	\pCert_i \subseteq_r \binom{U_i}{\lambda|U_i|} \times \binom{U_i}{\lambda|U_i|} \text{ such that } |\pCert_i| = 2^{K \cdot |U_i|},
\end{equation}
uniformly at random for two parameters $\lambda$ and $K$ that depend only on $\epsilon$ and that will be fixed later.

\begin{definition}[Compatibility Certificates]
    \label{def:compatibility-cert}
    For a given vector $\vec{a} \in \inputA$, we say that the tuple $(L_i, R_i) \in \pCert_i$ is an \emph{$\inputA^{(i)}$-certificate} for $\vec{a}_{U_i}$ (that is, $\vec{a}$ restricted to the indices in $U_i$) if
    \begin{equation*}
        \vec{a}^{\,-1}(0) \cap R_i = \emptyset \text{ and } \vec{a}^{\,-1}(2) \cap U_i \subseteq L_i.
    \end{equation*}
    Likewise we say that $(L_i, R_i)$ is a \emph{$\inputB^{(i)}$-certificate} for $\vec{b}_{U_i}$ if
    \begin{equation*}
        \vec{b}^{\,-1}(0) \cap L_i = \emptyset \text{ and } \vec{b}^{\,-1}(2) \cap U_i \subseteq R_i.
    \end{equation*}
    
    If for a certain pair $(\vec{a}, \vec{b}) \in \inputA \times \inputB$ there exist tuples $(L_i, R_i) \in \pCert_i$ such that $(L_i, R_i)$ is an $\inputA^{(i)}$-certificate for $\vec{a}$ and a $\inputB^{(i)}$-certificate for $\vec{b}$ for all $i \in [\ell]$, then we say that the tuple
    \[
        ((L_1,R_1),\ldots,(L_\ell,R_\ell)) \in \pCert_1 \times \cdots \times \pCert_\ell
    \]
    is a \emph{compatibility certificate} for the pair $(\vec{a}, \vec{b})$.
\end{definition}

We have defined our compatibility certificates to act as proofs that pairs of vectors are compatible. For instance, if $\vec{a}^{\,-1}(0) \cap R_i = \emptyset$ and $\vec{b}^{\,-1}(2) \cap U_i \subseteq R_i$ for every $i \in [\ell]$, this means that there is no index $j \in [d]$ such that $a_i = 0$ and $b_j = 2$, an event which would cause $\vec{a}$ and $\vec{b}$ to be incompatible. Formally, we prove the following claim.

\begin{claim}
    If $C$ is a compatibility certificate for $(\vec{a}, \vec{b}) \in \inputA \times \inputB$, then $\vec{a} - \vec{b} \in
    [-1:1]^d$.
\end{claim}
\begin{claimproof}
    Let $C = ((L_1,R_1),\ldots,(L_\ell,R_\ell)) \in \pCert_1 \times \cdots \times
    \pCert_\ell$ be a compatibility certificate for $(\vec{a}, \vec{b}) \in
    \inputA \times \inputB$ and assume for contradiction that $\vec{a} - \vec{b}
    \notin [-1:1]^d$. 
    
    This implies that there exists some index $j \in U_i \subseteq [d]$ such that $a_j - b_j \in \{-2,2\}$. Without loss of generality, suppose $a_j = 0$ and $b_j = 2$ as the other case is symmetric. Because $(L_i, R_i)$ is a $\inputB^{(i)}$-certificate for $\vec{b}$, we know that that $j \in R_i$. However, because $(L_i, R_i)$ is an $\inputA^{(i)}$-certificate for $\vec{a}$, we have $j \notin R_i$, a contradiction.
\end{claimproof}

\subsection{Number of Candidate Certificates}

Our next task is to set the parameters $\lambda$ and $K$ so that for a fixed compatible pair $(\vec{a}, \vec{b})$, we are likely to have at least one compatibility certificate in $\pCert_1\times \cdots \times \pCert_\ell$.

\begin{lemma}\label{lem:candidates}
    Suppose $(\vec{a}, \vec{b}) \in \inputA \times \inputB$ satisfies $\vec{a} - \vec{b} \in [-1 : 1]^n$. Conditioned on \eqref{eq:support}, if
    \begin{align*}
        \lambda & \in [\epsilon, 1/2 + \epsilon] \text{ and } \\ 
        K & \ge 
        2 H(\lambda) - H( 2(\lambda - \epsilon)) + o(1)
    \end{align*}
    with probability at least $2/3$ there exists a compatibility certificate in $\pCert_1\times \cdots \times \pCert_\ell$.
\end{lemma}
\begin{proof}
    Fix a compatible pair $(\vec{a}, \vec{b})$ as in the lemma statement and suppose that \eqref{eq:support} holds, so the elements of $a^{\,-1}(0)$, $a^{\,-1}(2)$, $b^{\,-1}(0)$ and $b^{\,-1}(2)$ are distributed evenly over our $\ell$-partition of $U$. 
    
    First, we consider the probability that for a certain $i \in [\ell]$ the random family of candidate certificates $\pCert_i$ contains a pair $(L, R)$ that is an $\inputA^{(i)}$-certificate for $\vec{a}_{U_i}$ and a $\inputB^{(i)}$-certificate for $\vec{b}_{U_i}$.   
    By \Cref{def:compatibility-cert}, for this to occur $L$ must contain the $\epsilon |U_i|$ elements of $\vec{a}^{\,-1}(2) \cap U_i$ and $L$ must be a subset of the $(1/2 + \epsilon) |U_i|$ elements of $U_i \setminus b^{\,-1}(0)$, leaving a choice of $(\lambda - \epsilon)|U_i|$ ``free'' elements from the $|U_i|/2$ elements of $\vec{a}^{\,-1}(1) \cap U_i$. $R$ is similar. So the number of choices for $(L, R)$ that satisfy both properties is
    \begin{displaymath}
        \binom{|U_i|/2}{(\lambda - \epsilon) |U_i|} \cdot
        \binom{|U_i|/2}{(\lambda - \epsilon) |U_i|}
        .
    \end{displaymath}
    On the other hand, the total number of pairs $(L,R) \subseteq U_i$ with $|L| =
    \lambda |U_i|,|R|=\lambda |U_i|$ is $\binom{|U_i|}{\lambda |U_i|}\cdot
    \binom{|U_i|}{\lambda |U_i|}$. Hence, the probability that
    a randomly selected $(L, R) \in \pCert_i$ is an $\inputA^{(i)}$-certificate for $\vec{a}_{U_i}$ and a $\inputB^{(i)}$-certificate for $\vec{b}_{U_i}$ is at least 
    \begin{displaymath}
	\binom{|U_i|/2}{(\lambda - \epsilon) |U_i|}^2 \bigg/ 
        \binom{|U_i|}{\lambda |U_i|}^2
    \end{displaymath}
    which simplifies to
    \[
        \Theta^*\left(2^{\left(H(2(\lambda - \epsilon)) - 2 H(\lambda)\right)|U_i|}\right)
    \]
    by Stirling's approximation \eqref{eq:stirling}. Hence by selecting
     \begin{displaymath}
        K \ge 2 H(\lambda) - H(2(\lambda - \epsilon)) + o(1)
    \end{displaymath}
    and $\lambda \in [\epsilon, 1/2 + \epsilon]$ we can ensure that $\pCert_i$ contains a pair $(L, R)$ that is an $\inputA^{(i)}$-certificate for $\vec{a}_{U_i}$ and a $\inputB^{(i)}$-certificate for $\vec{b}_{U_i}$ with arbitrarily high constant probability. This follows from the standard fact that
    \[
        \left(1 - \frac{1}{z}\right)^{rz} \leq e^{-r}
    \]
    for real numbers $z > 1$ and $r > 0$ if we let $1/z$ be the probability that
    a randomly selected $(L, R) \in \pCert_i$ is an $\inputA^{(i)}$-certificate for $\vec{a}_{U_i}$ and a $\inputB^{(i)}$-certificate for $\vec{b}_{U_i}$.
    
    Because this event occurs independently for each $\pCert_i$, $i \in [\ell]$, we can conclude that there exists a complete compatibility certificate with arbitrarily high constant probability by taking a union bound over the individual failure probabilities for $i \in [\ell]$.
\end{proof}

\subsection{Auxiliary Sets}

Our algorithm will look for a compatibility certificate that occurs in \emph{both} of two
\emph{auxiliary sets}, lists of $\inputA^{(i)}$- and $\inputB^{(i)}$-certificates with respect to certain vectors $\vec{x} \in [0:2]^{U_i}$. Specifically, for each $i \in [\ell]$ and for every vector $\vec{x} \in [0:2]^{U_i}$, we define the following two lists:
\begin{align}
    \label{eq:auxiliary-sets}
    \auxL^{(i)}(\vec{x}) \coloneqq & \{ (L, R) \in \pCert_i \mid (L, R) \text{ is an $\inputA^{(i)}$-certificate for } \vec{x} \} \text{ and }\\
    \auxR^{(i)}(\vec{x}) \coloneqq & \{ (L, R) \in \pCert_i \mid (L, R) \text{ is a $\inputB^{(i)}$-certificate for }\vec{x} \}.\nonumber
\end{align}
By definition, if $\vec{a}$ and $\vec{b}$ have a compatibility certificate in $\pCert_1 \times \dots \times \pCert_\ell$, then this compatibility certificate is a member of both $\auxL^{(1)} (\vec{a}_{U_1}) \times \cdots \times \auxL^{(\ell)}(\vec{a}_{U_\ell})$ and $\auxR^{(1)} (\vec{b}_{U_1}) \times \cdots \times \auxR^{(\ell)}(\vec{b}_{U_\ell})$.

\begin{claim}\label{clm:data-structures}
    For each $i \in [\ell]$ we can enumerate the lists $\auxL^{(i)}(\vec{x})$ and $\auxR^{(i)}(\vec{x})$ for \emph{every} vector $\vec{x} \in [0:2]^{U_i}$ in time $\Oh(d \cdot 12^{d/\ell})$.
\end{claim}
\begin{proof}
	Fix $i \in [\ell]$ and perform a brute-force enumeration: for every
        \[
            (\vec{x}, (L, R)) \in [0:2]^{U_i} \times \pCert_i,
        \]
        add $(L, R)$ to $\auxL^{(i)}(\vec{x})$ if $(L, R)$ is an $\inputA^{(i)}$-certificate for $\vec{x}$. Checking whether $(L, R)$ is an $\inputA^{(i)}$-certificate takes time $O(|U_i|) = O(d / \ell) = O(d)$, and we perform at most $\Oh(3^{d / \ell} \cdot 4^{d / \ell}) = \Oh(12^{d / \ell})$ checks since $|[0:2]^{U_i}| \leq 3^{d / \ell}$ and $|\pCert_i| \leq |\binom{U_i}{U_i / 2}|^2 \leq 2^{2 |U_i|} = 4^{d / \ell}$. Computing $\auxR^{(i)}(\vec{x})$ for $\vec{x} \in [0:2]^{U_i}$ is similar.
\end{proof}

\begin{lemma}\label{lem:ca}
    For any fixed $\vec{a} \in \inputA$ and $\vec{b} \in B$, conditioned on \eqref{eq:support} it holds that:
    \begin{align*}
		|\auxL^{(1)} (\vec{a}_{U_1})|\cdots |\auxL^{(\ell)}(\vec{a}_{U_\ell})| &\leq 2^{c_0 \cdot d} \text{ and } \\
            |\auxR^{(1)} (\vec{b}_{U_1})|\cdots |\auxR^{(\ell)}(\vec{b}_{U_\ell})| &\leq 2^{c_0 \cdot d}
    \end{align*}
    with probability at least $2/3$, where
    \begin{equation}
        c_0 \coloneqq H\left( \frac{\lambda-\epsilon}{1-\epsilon}
                \right)(1-\epsilon) + H\left( \frac{\lambda}{1/2 + \epsilon}
                \right)(1/2 + \epsilon)
                - 2H(\lambda)
        + K + o(1).
        \label{eq:ca}
    \end{equation}
\end{lemma}
\begin{proof}
    Suppose that \eqref{eq:support} holds, so that the elements of $\vec{a}^{\,-1}(0)$, $\vec{a}^{\,-1}(2)$, $\vec{b}^{\,-1}(0)$ and $\vec{b}^{\,-1}(2)$ distribute ``evenly'' over our $\ell$-partition of $U$. Our first step is to prove that for any $i \in [\ell]$ and any $\vec{x} \in [0:2]^{U_i}$ with $|\vec{x}^{\,-1}(0)| = (1/2 - \epsilon) |U_i|$ and $|\vec{x}^{\,-1}(2)| = \epsilon |U_i|$ it holds that
    \begin{align}\label{eq:lemca}
		\Ex\left[|\auxL^{(i)}(\vec{x})|\right] \le 2^{c_0 \cdot |U_i|}.
    \end{align}
    By definition $|\auxL^{(i)}(\vec{x})|$ is equal to the number of sets $(L,R) \in \pCert_i$ such
    that $\vec{x}^{\,-1}(2) \subseteq L$ and $\vec{x}^{\,-1}(0) \cap R = \emptyset$. Since sets $(L,R) \in \pCert_i$ are drawn uniformly at random, the
    probability that $(L,R)$ is an $\inputA^{(i)}$-certificate for $\vec{x}$ is equal to
    \begin{displaymath}
        \Prob{(L,R) \text{ is $\inputA^{(i)}$-certificate for } \vec{x} } =
        \binom{(1-\epsilon)|U_i|}{(\lambda -
		\epsilon)|U_i|}\binom{|U_i|}{\lambda |U_i|}^{-1}
        \cdot
        \binom{(1/2 + \epsilon)|U_i|}{\lambda|U_i|}\binom{|U_i|}{\lambda |U_i|}^{-1}.
    \end{displaymath}
    Here, the first term comes from dividing the number of ways to choose $L$ such that $\vec{x}^{\,-1}(2) \subseteq L$ by the total number of ways to choose $L$, and the second term comes from dividing the number of ways to choose $R$ such that $\vec{x}^{\,-1}(0) \cap R = \emptyset$ from the total number of ways to choose $R$.
    
    By linearity of expectation, applying Stirling's approximation \eqref{eq:stirling} to bound the binomials, we have:
    \begin{displaymath}
		\Ex[|\auxL^{(i)}(\vec{x})|] \le 
        2^{\left( H\left( \frac{\lambda-\epsilon}{1-\epsilon}
                \right)(1-\epsilon) + H\left( \frac{\lambda}{1/2 + \epsilon}
                \right)(1/2 + \epsilon)
                - 2H(\lambda) + o(1)
        \right) |U_i|} \cdot 2^{K \cdot |U_i|},
    \end{displaymath}
    which is precisely \eqref{eq:lemca}.
    
    Applying Markov's inequality to \eqref{eq:lemca} yields
    \[
        |\auxL^{(i)}(\vec{x})| \geq (6\ell) \cdot 2^{c_0 \cdot |U_i|}
    \]
    with probability at most $1/(6\ell)$, so it follows that
    \[
        |\auxL^{(i)}(\vec{a}_{U_i})| \leq (6\ell) \cdot 2^{c_0 \cdot |U_i|}
    \]
    for \emph{all} $i \in [\ell]$ with probability at least $5/6$ by a union bound. A symmetric argument leads to the conclusion that
    \[
        |\auxR^{(i)}(\vec{b}_{U_i})| \leq (6\ell) \cdot 2^{c_0 \cdot |U_i|}
    \]
    for \emph{all} $i \in [\ell]$ with probability at least $5/6$. Taking a final union bound implies that \emph{both} events occur with probability at least $2/3$.
\end{proof}

\subsection{Algorithm for Compatibility Testing}

We are now ready to present our algorithm for compatibility testing and prove \Cref{thm:compatibility} (see Algorithm~\ref{alg:2ov}). First, we draw a random family of partial certificates $\pCert_i$ for each $i \in [\ell]$ as in \eqref{eq:pCertSelection}. Next, we build auxiliary sets $\auxL^{(i)}(\vec{x})$ and $\auxR^{(i)}(\vec{x})$ for every $i \in [\ell]$ and every vector $\vec{x} \in [0:2]^{U_i}$.

The core of the algorithm is constructing a set that for every vector $\vec{a}
\in \inputA$ contains every tuple in $\auxL^{(1)} (\vec{a}_{U_1}) \times \cdots \times \auxL^{(\ell)}(\vec{a}_{U_\ell})$. If, during the enumeration of this set for a particular $\vec{a} \in \inputA$, we find that the set contains more than $2^{c_0 d}$ elements, we break the loop and continue with the next element of $\inputA$. We do the same for every $\vec{b} \in \inputB$ and check if the two lists contain a common element; if so, this is a compatibility certificate.

\begin{algorithm}[ht!]
    \DontPrintSemicolon
    \nlnonumber
    \textbf{function} $\textsf{Compatibility}(\inputA,\inputB)$:\\
    Randomly partition $[d] = U_1 \sqcup \ldots \sqcup U_\ell$\label{ln:pre1} \tcp*{see~Sec~\ref{sec:preprocessing}}
    Draw sets $\pCert_1,\ldots,\pCert_\ell$\label{ln:pre2} \tcp*{see~\eqref{eq:pCertSelection}}
	Construct $\auxL^{(1)}(\vec{x}),\ldots,\auxL^{(\ell)}(\vec{x}),\auxR^{(1)}(\vec{x}),\ldots,\auxR^{(\ell)}(\vec{x})$ for all $\vec{x} \in [0:2]^{U_i}$\label{ln:auxsets} \tcp*{see~\eqref{eq:auxiliary-sets}}
    Let $\textsf{Candidates} = \{\}$\\
	\ForEach{$\vec{a} \in \inputA$}{
		\ForEach(\tcp*[f]{break if more than $2^{c_0 \cdot d}$}){$C \in
            \auxL^{(1)}(\vec{a}_{U_1}) \times \cdots \times
		\auxL^{(\ell)}(\vec{a}_{U_\ell})$\label{ln:itera}}{
            Add $C$ to \textsf{Candidates}.
        }
    }
    \ForEach{$\vec{b} \in \inputB$}{
		\ForEach(\tcp*[f]{break if more than $2^{c_0 \cdot d}$}){$C \in
            \auxR^{(1)}(\vec{b}_{U_1}) \times \cdots \times
		\auxR^{(\ell)}(\vec{b}_{U_\ell})$\label{ln:iterb}}{
            \If{$C \in$ \textsf{Candidates}}{\Return compatible pair}
        }
    }
    \Return no compatible pair
	\caption{Outline of our algorithm for compatibility testing. Implementation details omitted from this figure are specified in the proof of~\cref{lem:2ov-with-constants}.}
    \label{alg:2ov}
\end{algorithm}


\begin{lemma}[Correctness and Runtime of Algorithm \ref{alg:2ov}]\label{lem:2ov-with-constants}
    Given $\inputA,\inputB \subseteq [0:2]^d$ satisfying the conditions in \Cref{thm:compatibility}, Algorithm \ref{alg:2ov} runs in time 
    \begin{displaymath}
        \left(2^{c_0 \cdot d} (|\inputA| + |\inputB|) + 12^{d/\ell}\right) \cdot d^{\Oh(\ell)}
    \end{displaymath}
    and recovers a compatible pair $(\vec{a}, \vec{b}) \in \inputA \times
    \inputB$ such that $\vec{a} - \vec{b} \in [-1:1]^d$ with probability $d^{-O(\ell)}$ if such a pair exists.
\end{lemma}
\begin{proof}
    First, we consider runtime. Lines~\ref{ln:pre1} and~\ref{ln:pre2} take
    polynomial time. By~\cref{clm:data-structures}, the construction of
    auxiliary lists in Line~\ref{ln:auxsets} takes $\Oh(d \cdot 12^{d/\ell})$
    time. Finally, the interior loops in Lines~\ref{ln:itera}
    and~\ref{ln:iterb} repeat $2^{c_0 \cdot d}(|\inputA|+|\inputB|)$ times. The
    operations inside the interior loops take time $O(1)$ (adding $C$ to
    \textsf{Candidates}) and time $O(\log(2^{c_0 \cdot d}(|\inputA|+|\inputB|))) = \poly(d)$ (searching \textsf{Candidates}, if we sort \textsf{Candidates} ahead of time). Summing the runtime of each line completes the running time analysis. 
    
    It remains to prove correctness. Because the algorithm never returns a false positive, it remains to show that the algorithm recovers a compatible pair $(\vec{a}, \vec{b}) \in \inputA \times \inputB$ if one exists. 
    
    Fix such a pair. By \Cref{obs:support}, with probability $d^{-O(\ell)}$
    \eqref{eq:support} holds and the 0's, 1's, and 2's in $\vec{a}$ and
    $\vec{b}$ are evenly distributed between $U_1, \ldots, U_\ell$. Conditioned on this event, the number of tuples to check in Line~\ref{ln:itera} and Line~\ref{ln:iterb} is at most $2^{c_0 \cdot d}$ with probability $2/3$ by \Cref{lem:ca}. Likewise conditioned on \eqref{eq:support}, the set $\textsf{Candidates}$ contains a compatibility certificate for $\vec{a}$ and $\vec{b}$ with probability $2/3$ by \Cref{lem:candidates}. By a union bound, these events occur simultaneously with probability at least $1/3$.

    If this happens, the set $\textsf{Candidates}$ will contain every tuple $C \in \auxL^{(1)}(\vec{a}_{U_1}) \times \cdots \times \auxL^{(\ell)}(\vec{a}_{U_\ell})$, including the valid compatibility certificate(s). This compatibility certificate is also present in $\auxR^{(1)}(\vec{a}_{U_1}) \times \cdots \times \auxR^{(\ell)}(\vec{a}_{U_\ell})$ so it is recovered deterministically by the loop in Line~\ref{ln:iterb}.
\end{proof}

\begin{proof}[Proof of~\cref{thm:compatibility}]
    Set $\lambda = 2 \eps$ and $K = H(2 \epsilon) + o(1)$ and observe that the conditions required by~\cref{lem:candidates} are satisfied. Substituting these values into \eqref{eq:ca} yields
    \[
        c_0 = (1-\eps) H\left( \frac{\eps}{1-\eps} \right) + (1/2 + \epsilon)\cdot H\left(
            \frac{2 \eps}{1/2 + \epsilon}
        \right) -  H(2\epsilon) + o(1).
    \]
    Finally, choosing a large constant value for $\ell$ such as $\ell = 100/c_0$ and substituting into the runtime given by \Cref{lem:2ov-with-constants} gives the final runtime claimed by \Cref{thm:compatibility}.
\end{proof}

\bibliographystyle{alpha}
\bibliography{main}

\appendix

\section{ Proof of \texorpdfstring{\Cref{lem:prime-dist-main}}{} }
\label{apx:prime-dist-main}

The prime hashing technique rests on the following lemma that follows from the prime number theorem:
\begin{lemma}[Folklore]
    \label{lem:prime-hashing} 
        For any sufficiently large positive integer $r$ and any positive integer $z$, a prime $\bp \sim [r : 2r]$ divides $z$ with probability at most $\log_2(z)/r$.
\end{lemma}

This lemma bounds the expected number of collisions observed when we hash any set of integers over a sufficiently large, randomly chosen prime. Recall our prime distribution lemma:

\primedist*
\begin{proof}
    Let $G$, $Y$, $p_{max}$, $\bp$, and $\bm{R}$ be defined as in the lemma statement, and define the set of \emph{colliding pairs}
    \[
        \bm{P} \coloneqq \left\{ (g_1, g_2) \in G \times G \; | \; g_1 \neq g_2, \bp \text{ divides } |g_1 - g_2| \right\}.
    \]

    By \Cref{lem:prime-hashing}, we have that for any pair of distinct elements $(g_1, g_2) \in G \times G$,
    \[
        \Prx_{
        \substack{\bm{p} \text{ prime} \\
        \sim [p_{max} \;:\; 2p_{max}]} }\left[\bm{p} \; \text{divides} \; |g_1 - g_2|\right]
        \leq \frac{\log_2(|g_1 - g_2|)}{p_{max}} 
        \leq \frac{\log_2(\diam(G))}{p_{max}} 
        < \frac{n^{c}}{p_{max}}.
    \]
    Thus we have
    \[
        \Ex_{\bp}[|\bm{P}|] = O(|G|^2) \cdot \frac{n^c}{p_{max}}
    \]
    by linearity of expectation. Applying Markov's inequality and letting the $O$-notation absorb the constant, we conclude that
    \begin{equation}
        |\bm{P}| = O(|G|^2) \cdot \frac{n^{c}}{p_{max}}
        \label{eq:P-cardinality}
    \end{equation}
    with constant probability. We proceed to prove that if very few residue classes contain many elements of $G \pmod{\bp}$, \eqref{eq:P-cardinality} does \emph{not} hold. Thus many residue classes will contain many elements of $G \pmod{\bp}$ with constant probability.

    \begin{claim}
        Suppose that the set of residue classes \[
        \bm{R}' \coloneqq \left\{ r \in [\bp] \; \middle| \; \left| \left\{ g \in G \; \middle| \; g 
        \equiv r \pmod{\bp} \right\} \right| \geq \frac{|G|}{4p_{max}} \right\}
        \]
        has cardinality $|\bm{R}'| \leq 2 p_{max} \cdot n^{-(c+1)}$. Then \eqref{eq:P-cardinality} does not hold.
        \label{claim:Rprime-lb}
    \end{claim}
    \begin{claimproof}
    Suppose that $|\bm{R}'|$ is bounded as in the claim statement. Under this assumption, the residue classes in the set $[\bp] \setminus \bm{R}'$ contain fewer than 
    \begin{equation}
        \bp \cdot \frac{|G|}{4 p_{max}} \leq 2 p_{max} \cdot \frac{|G|}{4 p_{max}} = \frac{|G|}{2}
        \label{eq:G-in-Rprime-bound}
    \end{equation}
    elements of $G$, so at least $|G|/2$ elements fall into residue classes in $\bm{R}'$. 
    
    We will lower bound the number of colliding pairs in $\bm{R}'$, which is the cardinality of the set
    \[
        \bm{P}' \coloneqq \{ (g_1, g_2) \in G \times G \; | \; g_1 \neq g_2, \; \bp \text{ divides } |g_1 - g_2|, \text{ and } g_1, g_2 \; (\bmod{\;\bp}) \in \bm{R}'\}.
    \]

    We observe that $|\bm{P}'|$ approximates the squared $L_2$-norm of the frequency vector that counts the distribution of the elements of $G$ over the residue classes in $\bm{R}'$. Specifically, letting $G[r]$ denote $\{g \in G \; | \; g = r \pmod{\bp}\}$, we have
    \begin{equation}
        \label{eq:Pprime-cauchy-schwarz-bound}
        |\bm{P}'| = \sum_{r \in \bm{R}'} \frac{|G[r]| \cdot (|G[r]| - 1)}{2}
        > \left(\sum_{r \in \bm{R}'} \frac{|G[r]|^2}{2}\right) - |G|,
    \end{equation}
    where we use $|G|$ as an easy upper bound on $\sum_{r \in \bm{R}'} |G[r]|$. Thus we can lower bound $|\bm{P}'|$ using the following standard fact that follows from the Cauchy-Schwarz Inequality: for any $d$-dimensional vector $\vec{v}$, 
    \begin{equation}
        \label{eq:cauchy-schwarz}
        \sum_{i 
        \in [d]} v_i^2 = ||\vec{v}||_2^2 \geq \frac{||\vec{v}||_1^2}{d}.
    \end{equation}
    Applying \eqref{eq:cauchy-schwarz} to \eqref{eq:Pprime-cauchy-schwarz-bound} using the fact that $\sum_{r \in \bm{R}'} |G[r]| \geq |G|/2$ yields
    \[
        |\bm{P}'| \geq \frac{(|G| / 2)^2}{2 \cdot |\bm{R}'|} - |G| = \Omega\left(\frac{|G|^2}{|\bm{R}'|}\right) 
        \geq \Omega(|G|^2) \cdot \frac{n^{c + 1}}{p_{max}}.
    \]
    However, this contradicts \eqref{eq:P-cardinality} as $\bm{P}' \subseteq \bm{P}$.
    \end{claimproof}

    Because \eqref{eq:P-cardinality} holds with constant probability, we can conclude that the set of residue classes 
    \[
        \bm{R}' \coloneqq \left\{ r \in [\bp] \; \middle| \; \left| \left\{ g \in G \;
        \middle| \; g \equiv r \pmod{\bp} \right\} \right| \geq \frac{|G|}{4p_{max}} \right\}
    \]
    has cardinality at least $2 p_{max} \cdot n^{-(c+1)}$ with constant probability by \Cref{claim:Rprime-lb}.

    Condition on this event and consider the distribution of the elements of the set $Y \supseteq G$ over the residue classes in $\bm{R}'$. The size of the set
    \[
        \{ y \in Y \; | \; y \equiv \bm{r} \pmod{\bp} \}
    \]
    is trivially at most 
    \[
        \frac{|Y|}{|\bm{R}'|} < \frac{|Y|}{2 p_{max}} \cdot n^{c+1}
    \]
    in expectation over the choice of a random residue $\bm{r} \in \bm{R}'$, and at most twice this value for at least half of the elements of $\bm{R}'$ by Markov's inequality. 

    Thus the subset of residue classes that contain at least $|G|/4 p_{max}$ elements of $G$ and at most $n^{c+1} \cdot |Y|/p_{max}$ elements of $Y$ has cardinality at least
    \[
        \frac{|\bm{R}'|}{2} > \frac{p_{max}}{n^{c+1}}
    \]
    if \eqref{eq:P-cardinality} holds by \Cref{claim:Rprime-lb}, which happens with constant probability over the choice of $\bp$. This is exactly the set $\bm{R}$. This completes the proof of \Cref{lem:prime-dist-main}.
\end{proof}

\section{Proof of \texorpdfstring{\Cref{claim:C0-runtime-analysis}}{}}
\label{apx:C0-runtime-lemma-analysis}

\begin{claim}
    For every positive integer $d$ there exists a constant $\gamma > 0$, depending only on $d$, such that
    \[
        \frac{(d+1)^n}{(d+1)^{\frac{2n}{3(2d+1)}} (d!)^{\frac{4n}{3(2d+1)}}} \cdot 2^{\gamma n} < (2d + 1)^{n/2}.
    \]
    \label{claim:C0-runtime-analysis}
\end{claim}
\begin{claimproof}
    Recall that our goal is to prove that the quantity
    \begin{equation}
        \frac{(d+1)^n}{(d+1)^{\frac{2n}{3(2d+1)}} (d!)^{\frac{4n}{3(2d+1)}}}
        \label{eq:quantity-1}
    \end{equation}
    is smaller than $(2d+1)^{n/2}$ by a function exponential in $n$ for any integer $d \geq 1$.

    We begin by rewriting the denominator of \eqref{eq:quantity-1} in terms of the factorial $(d+1)!$ and factoring out the common $n$ in the exponent:
    \begin{equation*}
        \frac{(d+1)^n}{(d+1)^{\frac{2n}{3(2d+1)}} (d!)^{\frac{4n}{3(2d+1)}}}
        = \left(
        \frac{d+1}{((d+1)!)^{\frac{2}{3(2d+1)}} \left(\frac{(d+1)!}{d+1}\right)^{\frac{2}{3(2d+1)}}}
        \right)^n
    \end{equation*}

    Next, we recall the standard fact that
    \[
        \left(\frac{k}{e}\right)^k \leq k!
    \]
    for positive integer $k$, which is a variant of Stirling's formula. Setting $k = d+1$ and applying the inequality to \eqref{eq:quantity-1}, we have that for $d \geq 0$,

    \begin{align*}
        \left(
        \frac{d+1}{((d+1)!)^{\frac{2}{3(2d+1)}} \left(\frac{(d+1)!}{d+1}\right)^{\frac{2}{3(2d+1)}}}
        \right)^n &\leq 
        \left(
        \frac{d+1}{(\frac{d+1}{e})^{\frac{(d+1) \cdot 2}{3(2d+1)}} (\frac{d+1}{e})^{\frac{(d+1) \cdot 2}{3(2d+1)}} (d+1)^{-\frac{2}{3(2d+1)}}} 
        \right)^n \\
        &=
        \left(
        \frac{(d+1) \cdot e^{\frac{4d + 4}{6d + 3}}}{
            (d+1)^{\frac{4d+2}{6d+3}}
        }
        \right)^n \\
        &= \left(
        (d+1)^{1/3} \cdot e^{\frac{4d + 4}{6d + 3}}
        \right)^n
    \end{align*}

    Compare the base of this exponential function with that of the Meet-in-the-Middle runtime, which is $(2d+1)^{1/2}$. The derivative of 
    \[
        (d+1)^{1/3} \cdot e^{\frac{4d + 4}{6d + 3}}
    \]
    with respect to $d$ is strictly less than the derivative of 
    \[
        (2d+1)^{1/2}
    \]
    with respect to $d$, and the two quantities are equal at $d \approx 8.88$. Thus we conclude that the lemma holds for integers $d \geq 9$.

    It remains to show that the lemma holds for integers $1 \leq d < 9$, which is easily verified. The table below lists calculated upper bounds on \eqref{eq:quantity-1} for $1 \leq d < 9$.

    \begin{table}[ht]
    \centering
    \begin{tabular}{l|l|l}
    \hline
    \multicolumn{1}{|l|}{$d$} & $\eqref{eq:quantity-1} $  & \multicolumn{1}{l|}{$(2d+1)^{n/2}$} \\ \hline
    1   & $< 1.715^n$  & $> 1.732^n$  \\ \hline
    2   & $< 2.154^n$  & $> 2.236^n$  \\ \hline
    3   & $< 2.492^n$  & $> 2.645^n$  \\ \hline
    4   & $< 2.772^n$  & $3^n$        \\ \hline
    5   & $< 3.013^n$  & $> 3.316^n$  \\ \hline
    6   & $< 3.227^n$  & $> 3.360$    \\ \hline
    7   & $< 3.419^n$  & $> 3.872^n$  \\ \hline
    8   & $<3.595^n$   & $> 4.123^n$                            
    \end{tabular}
    \end{table} \qedhere
    \end{claimproof}
\end{document}